\documentclass{article}
\title{\textsc{Memory-dependent abstractions of stochastic systems through the lens of transfer operators\thanks{This paper was accepted for publication and presentation at the 2025 Hybrid Systems: Computation and Control conference (HSCC 2025).}}}
\author{Adrien Banse$^\star$ \and Giannis Delimpaltadakis$^\dagger$ \and Luca Laurenti$^\ddagger$ \and Manuel Mazo Jr.$^\ddagger$ \and Raphaël M. Jungers$^\star$}
\date{%
    $^\star$UCLouvain, Belgium\\%
    $^\dagger$Eindhoven University of Technology, Netherlands\\
    $^\ddagger$Delft University of Technology, Netherlands
}

\usepackage{amsmath}
\usepackage{amssymb}
\usepackage{color}
\usepackage{parskip}
\usepackage{tikz}
\usepackage{pdfpages}
\usepackage{fullpage}
\usepackage{amsthm}
\usepackage{bbm}
\usepackage{enumerate}
\usepackage{hyperref}
\usepackage{stmaryrd}
\usepackage{dsfont}
\usepackage{algorithm}
\usepackage{algpseudocode}

\newtheorem{theorem}{Theorem}
\newtheorem{corollary}{Corollary}
\newtheorem{lemma}{Lemma}
\newtheorem{proposition}{Proposition}

\theoremstyle{definition}
\newtheorem{definition}{Definition}
\newtheorem{remark}{Remark}
\newtheorem{assumption}{Assumption}
\newtheorem{example}{Example}
\newtheorem{problem}{Problem}

\begin{document}
    \maketitle
    \begin{abstract}
        With the increasing ubiquity of safety-critical autonomous systems operating in uncertain environments, there is a need for mathematical methods for formal verification of stochastic models. Towards formally verifying properties of stochastic systems, methods based on discrete, finite Markov approximations -- \emph{abstractions} -- thereof have surged in recent years. These are found in contexts where: either a) one only has partial, discrete observations of the underlying continuous stochastic process, or b) the original system is too complex to analyze, so one partitions the continuous state-space of the original system to construct a handleable, finite-state model thereof. In both cases, the abstraction is an approximation of the discrete stochastic process that arises precisely from the discretization of the underlying continuous process. The fact that the abstraction is Markov and the discrete process is not (even though the original one is) leads to approximation errors. Towards accounting for non-Markovianity, we introduce memory-dependent abstractions for stochastic systems, capturing dynamics with memory effects. Our contribution is twofold. First, we provide a formalism for memory-dependent abstractions based on transfer operators. Second, we quantify the approximation error by upper bounding the total variation distance between the true continuous state distribution and its discrete approximation. 
    \end{abstract}

    \section{Introduction}

Autonomous systems operating in uncertain environments are becoming ubiquitous, with applications ranging from autonomous driving, to robots in rescue missions, smart grids, smart buildings,~etc.~\cite{Alamir2022, Pairet2022}. Towards safe deployment of such autonomous systems, mathematical methods to formally verify if they meet prespecified requirements (e.g., on safety or performance) are needed \cite{Knight2002, Lee2016}. 

To mathematically analyze the aforementioned systems, while accounting for uncertainty, \emph{stochastic models} are often employed. In this context,  methods that are based on discrete approximations of stochastic systems, called \emph{abstractions}, have recently surged (see \cite{Lavaei2022} for a survey on abstractions for stochastic systems). Abstractions arise in two different contexts, that nevertheless present many mathematical similarities: a) one has access only to partial, discrete observations of the underlying original stochastic process (this is related to the work on \emph{Markov state models} \cite{Sarich2010,NielsenFackeldeyWeber2013,Prinz2011}); b) the original system is too complex to analyze, so one partitions its continuous state-space to construct a finite-state abstraction (this is related to the work on \emph{abstraction-based methods} \cite{Lavaei2022, Abate2010, Lavaei2020, Lahijanian2015, Figueiredo2024, Delimpaltadakis2024, Meng2023}). In both cases, the abstraction takes the form of a finite (sometimes robust) Markov chain.

In virtue of the above, the abstraction is an approximation of the discrete stochastic process that arises from the discretization of the original process. In particular, while the discrete process is not Markov, even though the original continuous process is (see Figure~\ref{fig:non_markov} for an example), for computational reasons, the abstraction is generally constructed to be Markov. This leads to approximation errors. To alleviate approximation errors due to non-Markovianity, we introduce memory-dependent abstractions of stochastic systems (so far, only memoryless abstractions have been proposed; see \emph{Related work} below). The introduction of memory aims precisely at capturing memory effects inherent in the non-Markovian discrete process.
%

\begin{figure}[ht!]
    \centering
    \includegraphics[width = 0.7\linewidth]{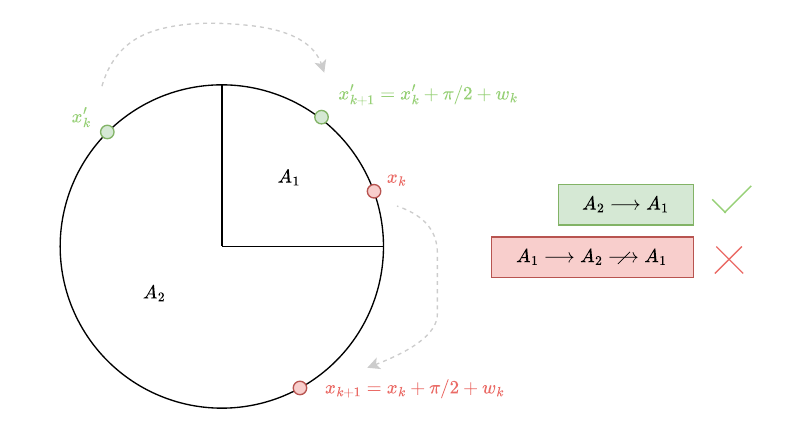}
    \caption{Loss of the Markov property. Consider a stochastic system defined by $x_{k+1} = x_k + \pi/2 + w_k \pmod{2\pi}$, where $w_k$ is a noise whose support is $[0, \pi/10]$. While the original system is Markovian, the discrete process tracking only in which region $x_k$ lies is not. Specifically, one can see that a state $x_{k}$ can jump from $A_2$ to $A_1$ with non-zero probability, i.e. $\mathbb{P}[x_{k+1} \in A_1 | x_{k} \in A_2] > 0$. However, when a larger memory is considered, we see that, if this state was initially in $A_1$, it cannot jump successively to $A_2$ and back to $A_1$, i.e. $\mathbb{P}[x_{k+2} \in A_1 | x_{k+1} \in A_2, x_k \in A_1] = 0$. The Markov property is therefore lost, since $\mathbb{P}[x_{k+1} \in A_1 | x_{k} \in A_2] \neq \mathbb{P}[x_{k+2} \in A_1 | x_{k+1} \in A_2, x_k \in A_1]$.}
    \label{fig:non_markov}
\end{figure}

\paragraph{Contributions}
In this paper, we develop memory-dependent abstractions of stochastic systems. Inspired by symbolic dynamics \cite{Lind1995}, we extend the state space of the stochastic system to create a lifted system, where each state represents an $\ell$-long sequence of states, $\ell$ being the considered memory. Further, akin to work on Markov state models \cite{Sarich2010,NielsenFackeldeyWeber2013,Prinz2011}, we employ transfer-operator theory and construct an $\ell$-memory abstraction, through Galerkin approximations of the lifted process's transfer operator. Critically, we provide an upper bound on the total variation distance between the distribution of the original continuous state system and its discrete approximation, enabling formal verification through the abstract model. Finally, we showcase through examples how memory increases approximation accuracy in various situations. This work therefore marks a significant step toward creating smart memory-dependent abstractions for the analysis and control of complex systems. 


\paragraph{Related work}
Memory-dependent abstractions have been developed for the analysis and control of deterministic systems \cite{Schmuck2014,Schmuck2015,Banse2023a,Banse2023}. However, to the best of our knowledge, such techniques have not been investigated for stochastic systems, where abstractions are memoryless (robust) Markov chains \cite{Abate2010, Lavaei2020, Lahijanian2015, Figueiredo2024, Delimpaltadakis2024, Meng2023}. Arguably, that is because incorporating memory in stochastic abstractions is fundamentally different than the deterministic case. For memory-dependent abstractions of deterministic systems, the \emph{domino rule} is employed, which, deeply rooted in determinism, is simply not applicable for stochastic systems (see Figure~\ref{fig:deterministic_vs_stochastic} for an explanation). Thus, extending memory-dependent abstraction techniques to stochastic systems presents a non-trivial challenge, and requires fundamentally different mathematical tools, which we develop here. 

Our work is also deeply related to and inspired by the work on Markov state models \cite{Sarich2010,NielsenFackeldeyWeber2013,Prinz2011}, which employs Galerkin approximations of the underlying Markov process's transfer operator, to build a finite approximation in a partial-observation scenario. Our contribution w.r.t. \cite{Sarich2010,NielsenFackeldeyWeber2013,Prinz2011} is twofold. First, we introduce memory to the discrete approximation, whereas only the memoryless case is studied in those works. Second, although we use intermediate results proven in \cite{Sarich2010}, our bounds - even in the memoryless case - fundamentally differ from those of the latter (see Remark~\ref{rem:difference_sarich} for more details). 

Finally, our work is related to partially observable Markov decision processes (POMDPs for short, see \cite{Suilen2025} for a survey). Indeed, the considered dynamical systems can be framed in the formalism of POMDPs, where the continuous state space of the original system and the discrete cells of the partition are the state and observation spaces, respectively, such as e.g. in \cite{5404339}. Although the loss of the Markov property for observations is a known phenomenon in POMDP literature \cite{guo2023sampleefficientlearningpomdpsmultiple}, our contribution departs from this literature in that we focus on problems that arise in the framework of (safety-critical) abstractions. That is, we study the loss of the Markov property when the state space is continuous and is approximated by cells corresponding to sets of states.

\begin{figure}[h!]
    \centering
    \includegraphics[width = 0.5\linewidth]{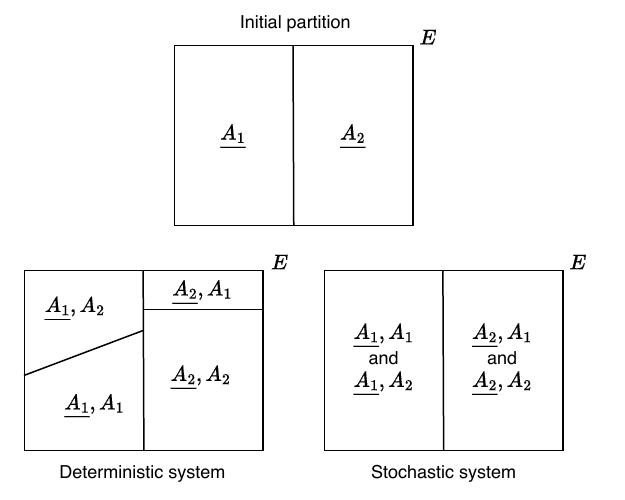}
    \caption{Domino rule for memory-dependent abstractions. Let $A_1$ and $A_2$ be two blocks of an initial partition on the state space $E$. $A_1$ and $A_2$ correspond to states of a 1-memory (memoryless) abstraction. Towards a 2-memory abstraction for deterministic systems, the domino rule proceeds as follows. The cell $A_1$ is divided into $A_1, A_1$ and $A_1, A_2$, the sets of states that are in $A_1$ and that will respectively be either in $A_1$ or in $A_2$ at the next timestep (the same happens to subdivide $A_2$). For stochastic systems, such division is not possible, as generally, even though the system might be in a specific state in $A_1$, it can visit any of $A_1$ or $A_2$ in the next step, due to stochasticity; in other words, there is no set of states in $A_1$ that deterministically visit either of $A_1$ and $A_2$ in the next step (similarly for the states initially in $A_2$).}
    \label{fig:deterministic_vs_stochastic}
\end{figure}

\paragraph{Outline} This paper is structured as follows. Section~\ref{sec:problem} defines the considered family of systems and the studied problem. Section~\ref{sec:preli} introduces key theoretical concepts, covering probability theory, transfer operators, and Galerkin approximations. Section~\ref{sec:abstraction} contains our overall method to abstract systems with $\ell$-memory Markov models. We present the lifted system on which relies our analysis, and we provide a mathematical framework that justifies our method. In Section~\ref{sec:bounds}, we provide total variation guarantees. Section~\ref{sec:numerical} provides numerical experiments, and Section~\ref{sec:conclu} concludes with discussions on the approach, its limitation and future directions.

\paragraph{Notations} Given any set $A$, $\chi_A$ is the \emph{indicator function} of $A$. $\mathbb{R}$ and $\mathbb{C}$ respectively denote the sets of \emph{real} and \emph{complex} numbers. The sets $\mathbb{R}_{> 0}$ and $\mathbb{R}_{\geq 0}$ respectively denote the set of positive and nonnegative real numbers. For $a \in \mathbb{C}$, $|a|$ denotes the \emph{modulus} of $a$. For a natural number $n$, the set $[n]$ denotes $\{1, \dots, n\}$. Let $E \subseteq \mathbb{R}^d$, we denote by $\mathcal{B}(E)$ the \emph{Borel set of $E$}, and the couple $(E, \mathcal{B}(E))$ forms a measurable space. All along this work we consider \emph{probability spaces} $(E, \mathcal{B}(E), \lambda)$, where $\lambda : \mathcal{B}(E) \to [0, 1]$ is a \emph{probability measure}. A set $(F, \langle \cdot, \cdot \rangle)$ is called a \emph{Hilbert space} if $\langle \cdot, \cdot \rangle$ is a \emph{dot product}. In this paper, given a measure $\mu$ on $E$, we consider the Hilbert space of \emph{square-integrable functions}, noted $L^2(\mu)$, defined as the set of functions $f : E \to \mathbb{R}$ such that $\int_{x \in E} (f(x))^2 \mu(\text dx) < + \infty$, and where the dot product is defined as $\langle f, g \rangle = \int_{x \in E} f(x)g(x) \mu(\text dx)$. The associated norm is defined as $\|f\|_2 = \langle f, f \rangle^{1/2}$. Let $L^1(\mu)$ be the set of functions such that $\int_{x \in E} |f(x)| \mu(\text dx) < + \infty$, we consider that $\mu$ is such that $L^2(\mu) \subseteq L^1(\mu)$, and we define $\|f\|_1 = \int_{x \in E} |f(x)| \mu(\text dx)$. Let $P : L^2(\mu) \to L^2(\mu)$ be an operator, the \emph{operator norm} of $P$ is defined as $\|P\|_p = \sup_{f \in L^2(\mu) : \|f\|_p \leq 1} \|Tf\|_p$, for $p = 1, 2$. Finally, given a measurable space $(E, \mathcal{F})$, where $\mathcal{F}$ is any $\sigma$-algebra, and given two probability measures $\mu$ and $\nu$ on $(E, \mathcal{F})$, the \emph{total variation distance} between $\mu$ and $\nu$ is defined as $\mathsf{TV}(\mu, \nu) = \sup_{A \in \mathcal{F}} |\sigma(A) - \nu(A)|$.


\section{Problem formulation} \label{sec:problem}

\subsection{System description}
In this work, we consider discrete-time stochastic dynamical systems defined as
\begin{equation} \label{eq:system}
\begin{cases}
    x_{k+1} \sim \tau(\cdot | x_k), \\
    x_0 \sim \lambda_0, \\
    y_k = h(x_k), 
\end{cases}
\end{equation}
where $x_k \in E$ is the \emph{state} of the system at time $k$ and the set $E \subseteq \mathbb{R}^d$ is the \emph{state space}. For all $x \in E$, $\tau(\cdot | x) : \mathcal{B}(E) \to [0, 1]$ is the \emph{transition kernel}. The probability measure $\lambda_0 : \mathcal{B}(E) \to [0, 1]$ is the \emph{initial measure}. $y_k \in F$ is the \emph{output} at time $k$, and the set $F$ is the \emph{output space}, which is assumed to be finite. Without loss of generality, we consider that $F = \{1, \dots, n\}$ throughout the paper. The function $h : E \to F$ is called the \emph{output function} and defines a \emph{partition} of the state space $E$. Indeed, let
\begin{equation} \label{eq:partition}
    A_i = \{x \in E : h(x) = i\}, 
\end{equation}
then the collection of sets $A_1, \dots, A_n$ is such that 
\begin{equation*}
\begin{aligned}
    \text{(covering)}& \quad \bigcup_{i = 1}^n A_i = E, \\
    \text{(pairwise disjoint)}& \quad \forall i \neq j: A_i \cap A_j = \emptyset.
\end{aligned}
\end{equation*}

A state sequence $x_0, x_1, \dots$ of system~\eqref{eq:system} is a realisation of the \emph{state stochastic process}, which is denoted by $(X^{\lambda_0}_k)_{k \geq 0}$. When the initial measure is clear from the context, we omit $\lambda_0$ from the notation and simply write $X_k$. Given the definition of system~\eqref{eq:system}, the probability measure associated to the state process is defined as
\begin{equation*}
\begin{aligned}
    \mathbb{P}[X_0 \in A_0] &= \lambda_0(A_0), \\
    \mathbb{P}[X_k \in A_k] &= \int_{x_{k-1} \in E} \dots \int_{x_0 \in E} \tau(A_k | x_{k-1}) \dots \tau(\text d x_1 | x_0) \lambda_0(\text dx_0),\\
    \mathbb{P}[X_0 \in A_0, \dots, X_k \in A_k] &= \int_{x_{k-1} \in A_{k-1}} \dots \int_{x_0 \in A_0} \tau(A_k | x_{k-1}) \dots \tau(\text d x_1 | x_0) \lambda_0(\text dx_0), 
\end{aligned}
\end{equation*}
where $A_0, \dots, A_k \in \mathcal{B}(E)$. We denote by $\lambda_k$ the probability measure induced by $\mathbb{P}[X_k \in \cdot]$. 

We now recall the definition of an \emph{invariant measure}.
\begin{definition}[Invariant measure]
    A measure $\mu : \mathcal{B}(E) \to [0, 1]$ is said to be \emph{invariant} for the system~\eqref{eq:system} if, for all $A \in \mathcal{B}(E)$, it holds that
    \begin{equation*}
        \int_{x \in E} \tau(A | x) \mu(\text dx) = \mu(A).
    \end{equation*}
\end{definition}

\begin{assumption}[Existence and uniqueness]\label{ass:ex_un}
    System~\eqref{eq:system} admits a unique invariant measure, denoted by $\mu$.
\end{assumption}
\begin{assumption}[Ergodicity]\label{ass:ergodic}
    System~\eqref{eq:system} converges in total variation to the invariant measure, that is
    \begin{equation*}
        \lim_{k \to \infty} \mathsf{TV}(\lambda_k, \mu) = 0.
    \end{equation*}
\end{assumption}
Assumptions \ref{ass:ex_un} and \ref{ass:ergodic} are standard in the literature \cite{Hairer2006ErgodicPO, Katok1995}, and commonly met in many cases of practical interest.
We leave the extension to systems that do not satisfy Assumption~\ref{ass:ex_un} and Assumption~\ref{ass:ergodic} as future work. 

The output sequence $y_0, y_1, \dots$ of system~\eqref{eq:system} is a realisation of the \emph{output stochastic process}, denoted by $(Y^{\lambda_0}_k)_{k \geq 0}$. Again, we omit $\lambda_0$ when it is clear from the context. The output process is defined as 
\begin{equation*}
    Y_k = h(X_k) \in \{1, \dots, n\}.
\end{equation*}
Observe that, although the continuous process $(X_k)_{k \geq 0}$ is a Markov process, the discrete process $(Y_k)_{k \geq 0}$ is generally non-Markovian \cite{Banse2023, Sarich2010}, that is 
\begin{equation*}
\begin{gathered}
    \mathbb{P}[Y_{k+1} = i_{k+1} | Y_k = i_k, \dots, Y_0 = i_0] \neq \mathbb{P}[Y_{k+1} = i_{k+1} | Y_k = i_k]. 
\end{gathered}
\end{equation*}
We invite the reader to see Figure~\ref{fig:non_markov} for an illustrative example of this phenomenon.




\subsection{Problem statement}

In this paper, given an infinitely long output sequence $\{y_i\}_i$, we aim at approximating the continuous Markov process $(X_k)_{k \geq 0}$ with a \emph{discrete, $\ell$-memory Markov process}, denoted by $(\tilde{Y}_{\ell, k})_{k \geq 0}$. That is, we construct the discrete process $(\tilde{Y}_{\ell, k})_{k \geq 0}$ such that 
\begin{equation*}
\begin{aligned}
    \mathbb{P}[\tilde{Y}_{\ell, k} = i_k | \tilde{Y}_{\ell, k-1} = i_{k-1}, \dots, \tilde{Y}_{\ell, 0} = i_0]
    = \mathbb{P}[\tilde{Y}_{\ell, k} = i_k | \tilde{Y}_{\ell, k-1} = i_{k-1}, \dots, \tilde{Y}_{\ell, k-\ell} = i_{k-\ell}].
\end{aligned}
\end{equation*}
In the literature, $\ell$-memory Markov processes are also sometimes referred to as \emph{Markov chains with memory $\ell$} (see e.g. \cite{Wu2017}), and constitute Markov chains in a lifted state space.

More precisely, we show that the discrete process $(\tilde{Y}_{\ell, k})_{k \geq 0}$ induces a probability measure $\tilde{\lambda}_{\ell, k}$ on the state space $E$, and we address the following problem. 
\begin{problem} \label{prob:tv} 
    Given an infinitely long output sequence $\{y_i\}_i$ of system \eqref{eq:system}, construct an $\ell$-memory Markov process $(\tilde{Y}_{\ell, k})_{k \geq 0}$ and compute the distance
    \begin{equation} \label{eq:tv_error}
        \mathsf{TV}(\lambda_k, \tilde{\lambda}_{\ell, k}),
    \end{equation}
    where $\tilde{\lambda}_{\ell, k}$ is the probability measure on the state space $E$ induced by $(\tilde{Y}_{\ell, k})_{k \geq 0}$.
\end{problem}


\begin{remark} \label{rem:practice_sampling}
    In practice, as we employ ergodicity of system \eqref{eq:system} (see Assumption~\ref{ass:ergodic}), the given output sequence $\{y_i\}_i$ has to be sufficiently -- not infinitely -- long, so that the Markov process $(X_k)_{k \geq 0}$ has almost reached its steady state (in the sense that $\mathsf{TV}(\lambda_k, \mu) < \varepsilon$, for small $\varepsilon$). See also Remark \ref{rem:computing steady state}. $\hfill \blacksquare$
\end{remark}

Observe that, through $\tilde{\lambda}_{\ell,k}$ and $\mathsf{TV}(\lambda_k, \tilde{\lambda}_{\ell, k})$, one can derive bounds on probabilistic properties of system \eqref{eq:system} (e.g., bounds on the probability of the state $X_k$ landing on some unsafe set, in some finite horizon). Our work is motivated by two distinct settings, which nevertheless present many mathematical similarities:

\textbf{Case 1 - Partially observable systems.} We would like to analyze properties of the underlying process $(X_k)_{k \geq 0}$, but can only observe (samples $y_k$ of) the output process $(Y_k)_{k \geq 0}$. Thus, through observing $(Y_k)_{k \geq 0}$, we construct a discrete, $\ell$-memory Markov approximation $(\tilde{Y}_{\ell,k})_{k \geq 0}$. Including memory is done precisely to capture non-Markovian effects inherent to the observed process $(Y_k)_{k \geq 0}$. As it will become evident both from the theory and the numerical examples, in certain cases, increasing memory $\ell$ leads to tighter approximations. Overall, this is related to earlier works on Markov state models \cite{Sarich2010,NielsenFackeldeyWeber2013,Prinz2011}.

\textbf{Case 2 - Finite abstractions.} System \eqref{eq:system} is fully observable (disregard the output process), but too complex to derive analytic results on its properties. Akin to standard abstraction methods, we discretize the state space to derive a finite partition $\{A_i\}_i$.
In contrast, through our method, we start from a coarser partition $\{A_i\}_i$, and, towards tighter approximations, the refinement is performed through increasing the abstraction's memory $\ell$; hence, the title of the paper. Alternatively, while standard abstraction methods, towards approximating $(X_k)_{k \geq 0}$, through partitioning the state space, implicitly approximate the non-Markovian process $(Y_k)_{k \geq 0}$ by a (1-memory) Markov process, we approximate $(Y_k)_{k \geq 0}$ by an $\ell$-memory Markov process $(\tilde{Y}_{\ell,k})_{k \geq 0}$, aiming precisely at capturing the non-Markov effects introduced exactly by partitioning the state-space in the first place

\paragraph{Approach} To address Problem~\ref{prob:tv}, we rely on the transfer operator of the stochastic system and its spectral properties. In particular, we define a lifted system that describes the evolution of a $\ell$-long sliding window of states $(x_k, \dots, x_{k+ \ell-1})$. We then construct a memory Markov process so that the transition matrix of the latter is a Galerkin approximation of the transfer operator of the lifted system. 
We derive two upper bounds on the total variation \eqref{eq:tv_error}. The first one consists of accumulation of projection errors, and increases with $k$. The second one is a consequence of the convergence of both models to the invariant measure, and decreases with $k$. For any memory $\ell$, the combination of these two bounds therefore provides a computable upper bound on the total distance.



\section{Preliminaries} \label{sec:preli}

As explained in Section~\ref{sec:problem}, our results rely on the theory of transfer operators, and their Galerkin approximation. This section formally introduces these concepts. 

\subsection{Probability theory}

We recall that $\mu$ denotes the unique invariant measure of system \eqref{eq:system}. We first give a definition of \emph{$\mu$-weighted probability density functions}.
\begin{definition}[$\mu$-weighted probability density function \cite{Prinz2011}] \label{def:mu-weighted}
    A function $v : E \to \mathbb{R}$ is called a \emph{$\mu$-weighted probability density function} if it is such that
    \begin{equation*}
        \int_{x \in E} v(x) \mu(\text dx) = 1,
    \end{equation*}
    and $v(x) \geq 0$ $\mu$-almost surely.
\end{definition}
The following remark gives an interpretation of $\mu$-weighted probability density functions with respect to usual probability density functions.

\begin{remark}
    Let $\sigma : \mathcal{B}(E) \to [0, 1]$ be any probability measure on the state space. If there exists $v : E \to [0, 1]$ such that  
    \begin{equation*}
        \sigma(A) = \int_{x \in A} v(x) \mu(\text dx),
    \end{equation*}
    for all $A \in \mathcal{B}(E)$ and if the latter is uniquely defined up to $\mu$-null sets, then $v$ is the so-called \emph{Radon-Nikodym derivative} of $\sigma$ with respect to the invariant measure, and is denoted by $\text d\sigma / \text d\mu$ (see e.g. \cite{Shiryaev2016} for more details). By Definition~\ref{def:mu-weighted}, $v$ is also a $\mu$-weighted probability density function. Note that, if it exists, the Radon-Nikodym derivative of $\sigma$ with respect to the Lebesgue measure, denoted by $p$, is known as the usual probability density function, and satisfies 
    \begin{equation*}
        \sigma(A) = \int_{x \in A} p(x) \text dx
    \end{equation*}
    for all $A \in \mathcal{B}(E)$. $\hfill \blacksquare$
\end{remark}

In this paper, since we assume the existence of a unique invariant distribution, we only work with $\mu$-weighted probability density functions and simply refer to them as probability density functions.

\begin{remark} \label{rem:invariant-is-1}
    Since $\mu(A) = \int_{x \in A} \mu(\text dx)$, the probability density function corresponding to the invariant measure is the constant function $\mathds{1}(x) = 1$ for all $x \in E$. $\hfill \blacksquare$
\end{remark}

Finally, consider two $\mu$-weighted probability density functions $v_1$ and $v_2$ that respectively correspond to two probability measures $\lambda_1$ and $\lambda_2$, then, the identity
\begin{equation*}
    \mathsf{TV}(\lambda_1, \lambda_2) = \frac{1}{2}\|v_1 - v_2\|_1 = \frac{1}{2} \int_{x \in E} |v_1(x) - v_2(x)| \mu(\text{d} x)
\end{equation*}
holds \cite{Hairer2006ErgodicPO}. 

\subsection{Transfer operator}

The \emph{transfer operator} corresponding to a transition kernel $\tau$ is defined as follows (see e.g. \cite{Sarich2010, Schuette2001}). 
\begin{definition}[Transfer operator] \label{def:transfer}
    Given a state space $E$ and a kernel $\tau$, the \emph{transfer operator} is the operator $T : L^2(\mu) \to L^2(\mu)$ such that
    \begin{equation*}
        \int_{x \in E} \tau(A | x) v(x) \mu(\text dx) =  \int_{y \in A} (Tv)(y) \mu(\text dy)
    \end{equation*}
    for all functions $v \in L^2(\mu)$ and all sets $A \in \mathcal{B}(E)$.
\end{definition}
If the measure $\tau(\cdot | x)$ admits a Radon-Nikodym derivative $t(\cdot | x) := \text d \tau(\cdot | x) / \text d\mu$ for all $x \in E$, then the transfer operator is explicitly defined as
\begin{equation} \label{eq:transfer}
    (Tv)(y) = \int_{x \in E} t(y | x) v(x) \mu(\text dx) 
\end{equation} 
for all functions $v \in L^2(\mu)$.

Intuitively, the operator $T$ propagates square-integrable functions of the state space in time, including probability density functions in $L^2(\mu)$. Let $v_0 = \text d\lambda_0 / \text d\mu \in L^2(\mu)$ be the probability density of $\lambda_0$ (the initial measure of System~\ref{eq:system}) then one can verify that 
\begin{equation} \label{eq:lambda_k}
    \lambda_k(A) := \mathbb{P}[X_k \in A] = \int_{x \in A} (T^k v_0)(x) \mu(\text dx).
\end{equation}
We therefore write $v_k := T^k v_0$ in the rest of this paper. Note that $v_k$ is the $\mu$-weighted probability density function of the measure $\lambda_k$. Furthermore, as pointed out in \cite{Sarich2010}, since $\mathds{1}$ is the invariant density (see Remark~\ref{rem:invariant-is-1}), $T$ satisfies $T\mathds{1} = \mathds{1}$, and $\mathds{1}$ is the only fixed point of $T$. Finally we define the spectrum of $T$.
\begin{definition}[Transfer operator spectrum] \label{def:spectrum}
    Let $e \in \mathbb{C}$ and $u \in L^2(\mu)$. If
    \begin{equation*}
        T u = eu, 
    \end{equation*}
    then $e$ and $u$ form a pair of \emph{eigenvalue} and \emph{eigenfunction} of $T$. All such pairs form the \emph{spectrum} of $T$. Moreover, a set of eigenfunctions $\{u_1, \dots, u_m\}$ is said to be \emph{orthonormal} if
    \begin{equation*}
        \langle u^\ell_i, u^\ell_j \rangle = 0
    \end{equation*}
    for all $i, j = 1, \dots, m$ such that $i \neq j$, and
    \begin{equation*}
        \|u_i^\ell\|_2 = \sqrt{\langle u^\ell_i, u^\ell_i \rangle} = 1
    \end{equation*}
    for all $i = 1, \dots, m$, where $\langle \cdot, \cdot \rangle$ is the inner product defined as 
    \begin{equation} \label{eq:dot_prod}
        \langle f, g \rangle = \int_{x \in E} f(x)g(x) \mu(\text dx).
    \end{equation}
\end{definition}
Note that the invariant density $\mathds{1}$ is an eigenfunction of $T$ with eigenvalue $1$, since $T\mathds{1} = \mathds{1}$.

\subsection{Galerkin methods} \label{subsec:galerkin}
In this work, inspired by works on Markov state models \cite{Sarich2010,NielsenFackeldeyWeber2013,Prinz2011}, we use a Galerkin method to approximate a lifted transfer operator. 

\begin{definition}[Projection operator] \label{def:projection}
    Let $(H, \langle \cdot, \cdot \rangle)$ be a Hilbert space, and let $\|\cdot\|$ be its associated norm. For a closed subspace $D \subset H$, the surjective map $Q: H \to D$ is an orthogonal projection onto $D$ if $Q^2 = Q$ and $\sup_{f \in H : \| f \| = 1} \|Qf\| = 1$. 
\end{definition}

Consider $(L^2(\mu), \langle \cdot, \cdot \rangle)$, the Hilbert space of square-integrable functions, where $\langle \cdot, \cdot \rangle$ is the usual dot product \eqref{eq:dot_prod}, together with a closed subspace $D \subset L^2(\mu)$ generated by a finite set of functions, that is 
\begin{equation*}
    D = \text{span}(\{\phi_1, \dots, \phi_n\}) 
\end{equation*}
for $n$ functions $\phi_i \in L^2(\mu)$, and let $Q : L^2(\mu) \to D$ be the unique projection operator defined above. Then the operator 
\begin{equation*}
    P := QTQ : D \to D
\end{equation*}
is called the \emph{Galerkin approximation} of $T$.  Since $P : D \to D$ is a finite-dimensional operator, it admits a \emph{matrix representation} $\mathbf{P} \in \mathbb{R}^{n \times n}$, defined as follows. For all $f \in D$, let $\mathbf{f} \in \mathbb{R}^n$ be the vector representation of $f$ in the basis $\{\phi_i\}_{i \in [n]}$, that is 
\begin{equation*}
    f = \sum_{i = 1}^n \mathbf{f}_i \phi_i.
\end{equation*}
Then, if $f' = Pf$ with vector representation $\mathbf{f}'$, it holds that 
\begin{equation*}
    \mathbf{f}' = \mathbf{P}^\top \mathbf{f}. 
\end{equation*}
In the following, we use boldface symbols to denote vector/matrix representations of functions/operators.

\section{Memory-dependent abstractions} \label{sec:abstraction}

In this section, we show how to approximate the true densities $v_k$ of system \eqref{eq:system} with the densities $\tilde{v}_{\ell, k}$ obtained from an $\ell$-memory Markov abstraction, itself derived by a Galerkin approximation of the transfer operator of a lifted system. Finally, we prove the correctness of our approach by upper bounding $\mathsf{TV}(\lambda_k, \tilde{\lambda}_{\ell, k})$. 


\subsection{Overall method} \label{sec:method}

Our approach to approximate $v_k$ is summarized in Algorithm~\ref{alg:l-approx}. We recall that, in what follows, $\chi_A$ denotes the indicator function of set $A$. Also, we consider transition matrices of $\ell$-memory Markov models, denoted $\mathbf{P}_{\ell} \in \mathbb{R}^{n^\ell \times n^\ell}$, where each row and column of such a matrix is labeled by $\ell$-long sequences of outputs $i_0i_1 \dots i_{\ell-1}$, where $i_j \in \{1, \dots, n\}$. Moreover, the matrix is such that
\begin{equation*}
    i_1 \dots i_{\ell-1} \neq j_1 \dots j_{\ell-1}\implies(\mathbf{P}_{\ell})_{i_0 \dots i_{\ell-1}, j_1 \dots, j_{\ell}} = 0.
\end{equation*}
Indeed the state $i_0\dots i_{l-1}$ of the $\ell$-memory Markov model represents the event $Y_0 =i_0, \dots, Y_{\ell-1} = i_{\ell-1}$, and the state $j_1,\dots, j_\ell$  represents the event $Y_1 = j_1, \dots, Y_\ell = j_\ell$. Thus naturally there can be no transition from $i_0,\dots,i_{\ell-1}$ to $j_1,\dots,j_\ell$ if $i_1 \dots i_{\ell-1} \neq j_1 \dots j_{\ell-1}$. As a consequence, the matrix $\mathbf{P}_{\ell}$ contains only $n^{\ell+1}$ possibly non-zero entries. Each considered vector $\mathbf{v}^\ell \in \mathbb{R}^{n^\ell}$ also has entries labeled with $i_0\dots i_{\ell-1}$, such that the matrix-vector product $\mathbf{P}_\ell \mathbf{v}^\ell$ is well defined.

Algorithm~\ref{alg:l-approx} proceeds as follows:
\begin{enumerate}
    \item[1.] An $\ell$-memory Markov chain is built, based on the steady-state dynamics (the dynamics when on the invariant measure). Each entry $(\mathbf{P}_\ell)_{i_0 \dots i_{\ell-1}, i_1 \dots i_\ell}$ contains the probability to go to the blocks $A_{i_\ell}$ knowing the $\ell$ last blocks were $A_{i_0}$, $\dots$, $A_{\ell-1}$. It directly follows from \eqref{eq:memory_probas} that this matrix is stochastic.
    \item[2.] The initial probability vector on the $n^\ell$ output sequences $i_0 \dots i_{\ell-1}$ is computed. Again, it follows from \eqref{eq:alg_eq_2} that this vector sums to 1.
    \item[3.] This probability is propagated $k-\ell+1$ times with the $\ell$-Markov chain transition matrix $\mathbf{P}_\ell$. The vector now contains entries labeled $i_{k-\ell+1} \dots i_{k}$ containing the approximated output joint probability from time $k-\ell+1$ to time $k$.
    \item[4.] The joint probability is marginalized so that the vector $(\tilde{\mathbf{v}}_{\ell, k})_{i_k}$ contains the approximated probabilities at time $k$, from which one may compute $\tilde{v}_{\ell,k}$.
\end{enumerate}
\begin{algorithm}
\caption{Compute $\tilde{v}_{\ell, k}$, the $\ell$-memory approximation of $v_k$}
\label{alg:l-approx} 
\begin{algorithmic}[1]
    \State Compute the $\ell$-memory transition probabilities of the output process $(Y^\mu_k)_{k \geq 0}$ of system \eqref{eq:system}, initialized at the invariant measure
    \begin{equation}  \label{eq:memory_probas}
        (\mathbf{P}_\ell)_{i_0\dots i_{\ell-1}, i_1 \dots i_\ell} := \mathbb{P}\left[Y_\ell^\mu = i_\ell \middle| Y^\mu_0 = i_0, \dots, Y^\mu_{\ell-1} = i_{\ell-1}\right].
    \end{equation} 
    \State Compute the initial $\ell$-long joint probabilities
    \begin{equation} \label{eq:alg_eq_2}
        (\tilde{\mathbf{v}}^\ell_0)_{i_0 \dots i_{\ell-1}} := \mathbb{P}\left[Y^{\lambda_0}_0 = i_0, \dots, Y^{\lambda_0}_{\ell-1} = i_{\ell-1}\right]
    \end{equation} 
    \State Propagate the $l$-long joint probabilities with the $\ell$-memory Markov model 
    \begin{equation} \label{eq:alg_eq_3}
        \tilde{\mathbf{v}}_{k - \ell + 1}^\ell = \left(\mathbf{P}_\ell^{k-\ell+1}\right)^\top \mathbf{\tilde{v}}^\ell_0
    \end{equation}
    \State Marginalize
    \begin{equation*}
        (\mathbf{\tilde{v}}_{\ell, k})_{i_k} := \sum_{i_{k - \ell + 1} = 1}^n \dots \sum_{i_{k-1} = 1}^n (\tilde{\mathbf{v}}^\ell_{k+\ell-1})_{i_{k-\ell+1} \dots i_{k-1} i_k}
    \end{equation*}
    \Return $\tilde{v}_{\ell, k}(x_k) = \sum_{i_k = 1}^n \dfrac{(\mathbf{\tilde{v}}_{\ell, k})_{i_k}}{\mathbb{P}[Y^\mu_k = i_k]} \chi_{A_{i_k}}(x_k)$. 
\end{algorithmic}
\end{algorithm}
The returned function in Algorithm~\ref{alg:l-approx} is a piecewise constant $\mu$-weighted probability density function, denoted $\tilde{v}_{\ell, k}$. It is such that 
\begin{equation*}
    \tilde{\lambda}_{\ell, k}(A) = \int_{x \in A} \tilde{v}_{\ell, k}(x) \mu(\text dx) \approx \int_{x \in A} v_k(x) \mu(\text dx) = \lambda_k(A)
\end{equation*}
for all $A \in \mathcal{B}(E)$, where $v_k = T^k v_0$ is the true $\mu$-weighted probability density function at time $k$, and $\approx$ denotes that we use $\tilde{\lambda}_{\ell, k}$ (or $\tilde{v}_{\ell, k}$) as an approximation of $\lambda_k$ (resp. $v_k$). This will be further explained in the next subsections.

\begin{remark}\label{rem:computing steady state}
    In practice, computing the invariant output probabilities \eqref{eq:memory_probas} can be done in at least two ways. Either one samples a sufficiently large number of ($\ell+1$)-long output traces initialized at the invariant distribution. Or, employing ergodicity and Birkhoff's theorem (see e.g. \cite{Shiryaev2016}), one samples a very large output trace initialized at any initial distribution. $\hfill \blacksquare$
\end{remark}

In the rest of this section, we introduce the mathematical formalism surrounding the construction of the abstraction.

\subsection{Lifted system}

Our approach is based on the study of the lifted state process $(X_k, \dots, X_{k+\ell-1})_{k \geq 0}$ and output process $(Y_k, \dots, Y_{k+\ell-1})_{k \geq 0}$. In the following subsections, we show that the abstraction constructed in Algorithm~\ref{alg:l-approx} is a Galerkin approximation of the transfer operator of this lifted process.  In this section, we formally define it along with its invariant distribution. We then conclude by making the link with the original system~\ref{eq:system}.

The lifted system is defined as
\begin{equation} \label{eq:system_lift}
\begin{cases}
    (x_{k +1}, \dots, x_{k+\ell}) \sim \tau^{\ell}(\cdot | x_{k}, \dots, x_{k + \ell - 1}), \\
    (x_0, \dots, x_{\ell-1}) \sim \lambda_0^\ell, \\
    (y_k, \dots, y_{k+\ell-1}) = (h(x_k), \dots, h(x_{k+\ell-1})).
\end{cases}
\end{equation}
In the definition above, for all $A_1, \dots, A_\ell \in \mathcal{B}(E)$ and all $x_0, \dots, x_{\ell-1} \in E$, the lifted kernel $\tau^\ell$ is defined as 
\begin{equation} \label{eq:kernel_lift}
\begin{aligned} 
    \tau^\ell (A_1 \times \dots \times A_\ell | x_0, \dots, x_{\ell-1} ) 
    = \begin{cases}
        \tau(A_\ell | x_{\ell-1}) &\text{if } x_1 \in A_1, \dots, x_{\ell-1} \in A_{\ell-1}, \\
        0 &\text{otherwise}.
    \end{cases}
\end{aligned}
\end{equation}
For all sets $A_0, \dots, A_{\ell-1} \in \mathcal{B}(E)$, the initial measure $\lambda_0^\ell$ is defined as 
\begin{equation*}
\begin{aligned}
    \lambda_0^\ell(A_0 \times \dots \times A_{\ell-1})
     = \int_{x_0 \in A_0} \dots \int _{x_{\ell-1 \in A_{\ell-1}}} \tau(\text dx_{\ell-1} | x_{\ell-2}) \dots \tau(\text dx_1 | x_0) \lambda_0(\text dx_0).
\end{aligned}
\end{equation*}


Owing to Assumption~\ref{ass:ex_un}, the lifted system admits a unique invariant measure $\mu^\ell$ (see e.g. \cite[Equation (4.1)]{Hairer2006ErgodicPO}), defined as 
\begin{equation}\label{eq:invariant_lift}
\begin{aligned} 
    \mu^\ell(A_0 \times \dots \times A_{\ell-1}) 
    = \int_{x_0 \in A_0} \dots \int_{x_{\ell-1} \in A_{\ell-1}} \tau(\text dx_{\ell-1} | x_{\ell-2}) \dots \tau(\text dx_1 | x_0) \mu(\text dx_0).
\end{aligned}
\end{equation}


Lifted system~\eqref{eq:system_lift} admits a transfer operator $T_\ell$, according to Definition~\ref{def:transfer}. The initial measure $\lambda^\ell_0$ admits a $\mu^\ell$-weighted probability density function, denoted $v^\ell_0(x_0, \dots, x_{\ell-1})$, which is a \emph{joint probability density function} on the first $\ell$ states. These joint densities are propagated with the lifted transfer operator, and, for all $k \geq \ell-1$,
\begin{equation*}
    v^\ell_{k-\ell+1}(x_{k-\ell+1}, \dots, x_k) = (T_\ell^{k-\ell+1}v_0)(x_{k-\ell+1}, \dots, x_k)
\end{equation*}
is the joint probability density on the states $x_{k-\ell+1}$ to $x_k$. The corresponding measure is denoted by $\lambda^\ell_{k - \ell + 1}$.

\begin{remark}[Notations]
    Study of joint measures and joint probability density functions are at the center of this work. We therefore draw the reader's attention on the fact that, all along the paper, we note joint measure (resp. $\mu^\ell$-weighted density) on $E^\ell$ with a superscript $\lambda^\ell$ (resp. $v^\ell$), whereas measures (resp. $\mu$-weighted densities) on $E$ are without any superscript $\lambda$ (resp. $v$). In contrast, superscripts on operators, e.g., $P^k$ or $T^k$, denote powers (or recursive applications of the operator).~$\hfill \blacksquare$
\end{remark}

\subsection{Abstraction}
In this section, we show that the transition matrix of the $\ell$-memory Markov chain constructed in Algorithm~\ref{alg:l-approx} corresponds to a Galerkin approximation of the transfer operator $T_\ell$. In particular, we specify the basis of functions with which we project $T_\ell$, and re-interpret Algorithm~\ref{alg:l-approx} in terms of functions and operators. Doing this will allow us to derive bounds in Section~\ref{sec:bounds} on the total variation distance between $v_k$ and the approximated function $\tilde{v}_{\ell, k}$ given by Algorithm~\ref{alg:l-approx}.

Given the output partition $A_1, \dots, A_n$ on the original state space $E$ (as defined in \eqref{eq:partition}), we consider the subspace of piecewise constant functions $D_{n}^\ell \subset L^2(\mu^\ell)$, defined as 
\begin{equation*}
    D^\ell_{n} := \text{span}\left(\left\{
        \psi_{i_1 \dots i_{\ell}}
    \right\}_{i_1, \dots, i_\ell \in [n]} \right), 
\end{equation*}
where 
\begin{equation} \label{eq:fix_denom}
    \psi_{i_1 \dots i_{\ell}}(x_1, \dots, x_{\ell})
    := 
    \frac{\chi_{A_{i_1}}(x_1) \dots \chi_{A_{i_\ell}}(x_\ell)}{\mathbb{P}[X^\mu_1 \in A_{i_1}, \dots, X^\mu_{\ell} \in A_{i_\ell}]}, 
\end{equation}
where we recall that $\chi_A$ denotes the indicator function of $A$. $D_{n}^\ell$ is therefore a set of piecewise constant functions on $E^\ell$.

In this work, we make the assumption that the denominator in \eqref{eq:fix_denom} is positive (as formally stated below in Assumption~\ref{ass:nonzero_denom}). We claim that this assumption is not restrictive for two main reasons. First, it holds in many practical cases such as unbounded noise. Second, it also suffices to assume that $\mathbb{P}[X^\mu_1 \in A_{i_1}, \dots, X^\mu_{\ell} \in A_{i_\ell}] = 0$ implies $\mathbb{P}[X^{\lambda_0}_1 \in A_{i_1}, \dots, X^{\lambda_0}_{\ell} \in A_{i_\ell}] = 0$. That is, zero measure steady state correspond to zero measure initial conditions. For the sake of brevity and simplicity, we leave this extension for further work.
\begin{assumption} \label{ass:nonzero_denom}
    The system~\eqref{eq:system} is such that its invariant measure $\mu$ satisfies 
    \begin{equation}
        \mathbb{P}[X^\mu_1 \in A_{i_1}, \dots, X^\mu_{\ell} \in A_{i_\ell}] > 0
    \end{equation}
    for all sequence $A_{i_1}, \dots, A_{i_\ell}$ of blocks of the output partition.
\end{assumption}

The following proposition shows that $\mathbf{P}_{\ell}$, the matrix built in Algorithm~\ref{alg:l-approx}, is the matrix representation of the Galerkin approximation of $T_\ell$. For the sake of readability, all proofs can be found in Appendix~\ref{sec:proofs}.
\begin{proposition} \label{prop:galerkin}
    Let $Q_\ell : L^2(\mu^\ell) \to D_n^\ell$ be a projection operator as defined in Definition~\ref{def:projection}, and let 
    \begin{equation*}
        P_\ell := Q_\ell T_\ell Q_\ell
    \end{equation*} 
    be the Galerkin approximation of $T_\ell$ on $D_{n}^\ell$. Then it holds that 
    \begin{equation} \label{eq:prop1_to_prove}
    \begin{aligned}
        P_\ell \psi_{i_0 \dots i_{\ell-1}} 
        = \sum_{i_\ell = 1}^n \mathbb{P}\left[Y_\ell^\mu = i_\ell \middle| Y^\mu_0 = i_0, \dots, Y^\mu_{\ell-1} = i_{\ell-1}\right] \psi_{i_1 \dots i_\ell}
    \end{aligned}
    \end{equation}
    for all $i_0, \dots, i_{\ell-1} \in [n]$. Therefore $\mathbf{P}_\ell$, as defined in \eqref{eq:memory_probas}, is the matrix representation of $P_\ell$.
\end{proposition}

We can therefore re-interpret Algorithm~\ref{alg:l-approx} through the lens of transfer operators and their Galerkin approximations. Figure~\ref{fig:method_summary} summarizes this interpretation. First, we compute $P_\ell$, the Galerkin approximation of $T_\ell$, whose matrix representation is $\mathbf{P}_\ell$, as computed in \eqref{eq:memory_probas} in Algorithm~\ref{alg:l-approx}. Second, we compute the piecewise approximate initial density
\begin{equation*}
\begin{aligned}
    v_0^\ell(x_0, \dots, x_{\ell-1}) \approx 
    \tilde{v}_0^\ell(x_0, \dots, x_{\ell-1}) = (Q_\ell v_0^\ell)(x_0, \dots, x_{\ell-1}), 
\end{aligned}
\end{equation*}
whose vector representation is $\tilde{\mathbf{v}}^\ell_0$, computed in \eqref{eq:alg_eq_2} in Algorithm~\ref{alg:l-approx}. Third, we approximate the joint density from $x_{k-\ell+1}$ to $x_k$ with
\begin{equation} \label{eq:approximate_lift_pdf}
\begin{aligned}
    v^\ell_{k - \ell + 1}(x_{k-\ell+1}, \dots, x_k) \approx \tilde{v}^\ell_{k - \ell + 1}(x_{k-\ell+1}, \dots, x_k) 
    = (P_\ell^{k - \ell + 1}v_0^\ell)(x_{k-\ell+1}, \dots, x_k). 
\end{aligned}
\end{equation}
Since $P_\ell = Q_\ell T_\ell Q_\ell$, it holds that 
\begin{equation*}
    P_\ell^{k - \ell + 1}v_0^\ell = P_\ell^{k - \ell + 1} (Q_\ell v_0^\ell) = P_\ell^{k - \ell + 1} \tilde{v}_0^\ell, 
\end{equation*}
and therefore that the vector representation of $\tilde{v}^\ell_{k - \ell + 1}$ is $\tilde{\mathbf{v}}^\ell_{k-\ell+1}$, as computed in \eqref{eq:alg_eq_3}. Finally, we marginalize $\tilde{v}^\ell_{k-\ell+1}(x_{k-\ell+1}, \dots, x_k)$ to get $\tilde{v}_{\ell, k}(x_k)$, whose vector representation in $D_n$ is $\tilde{\mathbf{v}}_{\ell, k}$. 


\begin{figure*}
    \includegraphics[width = \textwidth]{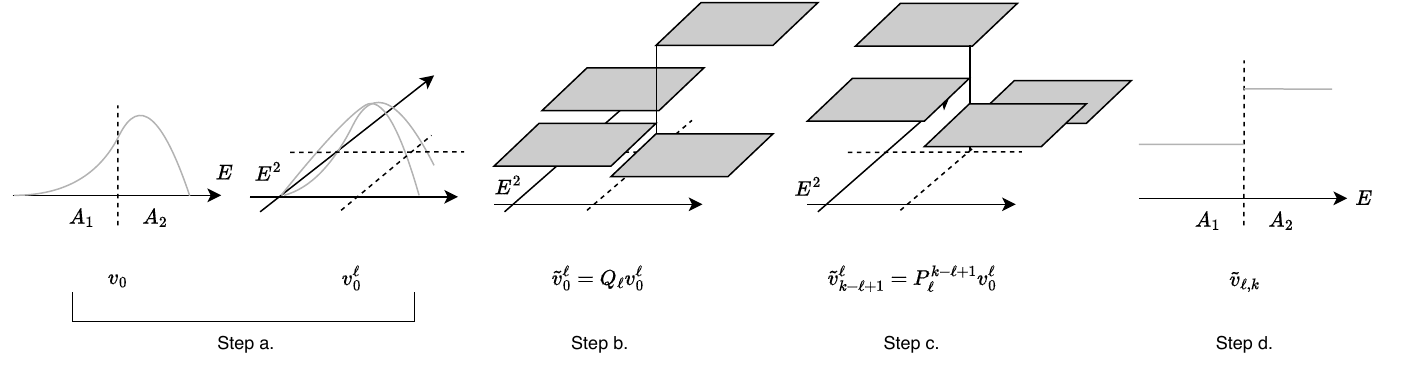}
    \caption{Summary of the method described in Algorithm~\ref{alg:l-approx}, from left to right. Step a) We consider the lifted process defined in \eqref{eq:system_lift}. Step b) We project the $x_0$-to-$x_{\ell-1}$ joint distribution on a finitely generated space $D_n^\ell$. Step c) We propagate it with the Galerkin approximation $P_\ell$ to get the approximate $x_{k-\ell+1}$-to-$x_k$ joint distribution. Step d) We marginalize it to retrieve the approximate density at $x_k$.}
    \label{fig:method_summary}
\end{figure*}

\section{Total variation guarantees} \label{sec:bounds}
In this section we upper bound $\mathsf{TV}(v_k, \tilde{v}_{\ell, k})$, the total variation between the true density $v_k$ and $\tilde{v}_{\ell, k}$, the approximated density given by Algorithm~\ref{alg:l-approx}, thereby providing formal guarantees on the correctness of our approach.
To derive our bounds, we require certain assumptions on the spectrum of the lifted transfer operator $T_\ell$, which are equivalent to those considered in  \cite{Sarich2010} and  formally defined in Assumption~\ref{ass:spectral}.

\begin{assumption} \label{ass:spectral}
    The lifted transfer operator $T_\ell$ admits $m$ real eigenvalues $e_{\ell, 1}, \dots, e_{\ell, m}$ such that $1 = e_{\ell, 0} > e_{\ell, 1} \geq \dots e_{\ell, m}$, together with an orthonormal set of eigenfunctions $\mathds{1} = u_0^\ell, u_1^\ell, \dots, u^\ell_m$ (see Definition~\ref{def:spectrum}). Moreover, for some $r_\ell < |e_{\ell, m}|$, all remaining eigenvalues $e_{\ell, i} \in \mathbb{C}$ for $i > m$ are such that $|e_{\ell, i}| < r_\ell$. Finally let $\Pi_\ell : L^2(\mu^\ell) \to L^2(\mu^\ell)$ be the operator defined by 
    \begin{equation}
        \Pi_\ell f^\ell = \sum_{i = 1}^m \langle f^\ell, u^\ell_i \rangle u_i^\ell. 
    \end{equation}
    We assume that $\Pi_\ell T_\ell = \Pi_\ell T_\ell \Pi_\ell$.
\end{assumption}

\begin{remark}
    Assumption~\ref{ass:spectral} is related but not the same as Assumption~\ref{ass:ex_un} and Assumption~\ref{ass:ergodic}. As discussed in \cite{Sarich2010}, sufficient conditions for Assumption~\ref{ass:spectral} to hold are \emph{reversibility} and \emph{sufficient ergodicity} (such as defined in \cite[Remark~2.1]{Sarich2010}), which are stronger assumptions than Assumption~\ref{ass:ex_un} and Assumption~\ref{ass:spectral}. As stated in \cite{Sarich2010,Kontoyiannis2011}, reversibility and sufficient ergodicity are natural and satisfied for a large class of dynamical systems.$\hfill \blacksquare$
\end{remark}

Similarly as in \cite{Sarich2010}, our bound relies on the quantity
\begin{equation} \label{eq:delta_lift}
    \delta_\ell := \max_{i = 1, \dots, m} \|Q_\ell u_i^\ell - u_i^\ell\|_2, 
\end{equation}
which quantifies the maximal projection error on the spectrum. 

Our total variation bound consists of two components. The first component increases with $k$ and arises from the cumulative projection errors, becoming more conservative as $k$ grows. On the other hand, the second component is characterized by the convergence of both the true density and the approximated one towards the invariant density $\mathds{1}$, and decreases with $k$. Unlike the first component, it is initially conservative but tightens progressively over time. The increasing and decreasing components are studied respectively in Theorem~\ref{thm:increasing} and Theorem~\ref{thm:decreasing} for the joint densities, and the final bound is given for the marginalized densities in Corollary~\ref{cor:final_bound}. 
We stress that both Theorem~\ref{thm:increasing} and Theorem~\ref{thm:decreasing} bound the same quantity. However, these bounds are complementary, as one is producing better bounds for small $k$, while the other for large $k.$

\begin{theorem}[Increasing] \label{thm:increasing}
    For any memory $\ell \geq 1$, horizon $k \geq \ell$ and initial joint density $v^\ell_{0} \in L^2(\mu^\ell)$, 
    let $\lambda^{\ell}_{k-\ell+1}$ and $\tilde{\lambda}^{\ell}_{k-\ell+1}$ be the joint measures respectively defined by 
    \begin{equation*}
    \begin{aligned}
        v^\ell_{k-\ell+1} &= T_\ell^{k-\ell+1}v_0^\ell, \\
        \tilde{v}^\ell_{k-\ell+1} &= P_\ell^{k-\ell+1}v_0^\ell, 
    \end{aligned}
    \end{equation*}
    and similarly for $\lambda^\ell_{k-\ell}$ and $\tilde{\lambda}^\ell_{k-\ell}$. Then, if Assumption~\ref{ass:spectral} is satisfied, it holds that 
    \begin{equation} \label{eq:def_theorem_measures}
    \begin{aligned}
        \mathsf{TV}(\lambda^\ell_{k - \ell + 1}, \tilde{\lambda}^\ell_{k - \ell + 1}) \leq \mathsf{TV}(\lambda^\ell_{k - \ell}, \tilde{\lambda}^\ell_{k - \ell}) + \frac{1}{2}\left(m e_{\ell, 1} \delta_\ell + r_\ell\right) e_{\ell, 1}^{k-\ell} \|v_0^\ell\|_2, 
    \end{aligned}
    \end{equation}
    where $\delta_{\ell}$ is defined in \eqref{eq:delta_lift}. 
\end{theorem}

\begin{theorem}[Decreasing] \label{thm:decreasing}
    For any memory $\ell \geq 1$, horizon $k \geq \ell$ and initial joint density $v^\ell_{0} \in L^2(\mu^\ell)$, let $\lambda^{\ell}_{k-\ell+1}$ and $\lambda^{\ell}_{k-\ell+1}$ be the joint measures respectively defined as in \eqref{eq:def_theorem_measures}. Then, if Assumption~\ref{ass:spectral} is satisfied, it holds that 
    \begin{equation*}
        \mathsf{TV}(\lambda^\ell_{k - \ell + 1}, \tilde{\lambda}^\ell_{k - \ell + 1}) \leq e_{\ell, 1}^{k - \ell + 1} \|v_0^\ell\|_2.
    \end{equation*}
\end{theorem}

Note that the following corollary is a result on the final marginalized densities, and not the joint densities such as in Theorem~\ref{thm:increasing} and Theorem~\ref{thm:decreasing}.
\begin{corollary} \label{cor:final_bound}
    For any memory $\ell \geq 1$, horizon $k \geq \ell$ and initial joint density $v^\ell_{0} \in L^2(\mu^\ell)$, let 
    \begin{equation*}
    \begin{aligned}
        \overline{\mathsf{TV}}_{\text{inc}} &:= \mathsf{TV}(\lambda_0^\ell, \tilde{\lambda}_0^\ell) +
        \left(m e_{\ell, 1} \delta_\ell + r_\ell\right)
        \frac{1-e_{\ell, 1}^{k-\ell+1}}{1 - e_{\ell, 1}}\|v_0^\ell\|_2, \\
        \overline{\mathsf{TV}}_{\text{dec}} &:= e_{\ell, 1}^{k-\ell+1}\|v_0^\ell\|_2, 
    \end{aligned}
    \end{equation*}
    where $\delta_{\ell}$ is defined in \eqref{eq:delta_lift}. Also, let $\lambda_k$ and $\tilde{\lambda}_{\ell, k}$ be the measures respectively defined by
    \begin{equation*}
    \begin{aligned}
        &v_k = T^k v_0, \\
        &\tilde{v}_{\ell, k} \text{ output of Algorithm~\ref{alg:l-approx}}.
    \end{aligned}
    \end{equation*}
    Then it holds that 
    \begin{equation*}
        \mathsf{TV}(\lambda_k, \tilde{\lambda}_{\ell, k}) \leq  
        \min \left\{ \overline{\mathsf{TV}}_{\text{inc}}, \overline{\mathsf{TV}}_{\text{dec}} \right\}. 
    \end{equation*}
\end{corollary}

\begin{remark} \label{rem:difference_sarich}
We should stress that, although we rely on similar tools, our bounds differ from those of \cite{Sarich2010}, even in the memoryless case ($\ell = 1$). Indeed the authors of \cite{Sarich2010} consider the operator norm
\begin{equation*}
    \|Q_1T_1^kQ_1 - (Q_1T_1Q_1)^k \|_2
\end{equation*}
as the error, which makes the assumption that the ground truth probability is $v_k = Q_1T_1^kQ_1v_0$. The setting of \cite{Sarich2010} therefore does not correspond to the setting of this paper, as we consider that the ground truth is $v_k = T_1^k v_0$. Taking into account this continuous ground truth makes our error larger, and consists in a supplementary technical challenge than a simple extension of \cite{Sarich2010}.$\hfill \blacksquare$
\end{remark}


\section{Numerical experiments} \label{sec:numerical}

In this section we motivate the method described in Section~\ref{sec:method} for the two cases described in Section~\ref{sec:problem}. In both cases, we consider the following dynamical system. 
\begin{example} \label{ex:2d}
    Consider the 2-dimensional linear system of the form 
    \begin{equation*}
    \begin{cases}
        x_{k+1} = A x_k + w_k, \\
        w_k \sim \mathcal{N}(m_w, \Sigma_w), \\
        x_0 \sim \mathcal{N}(m_0, \Sigma_0), 
    \end{cases}
    \end{equation*}
    with
    \begin{equation}
        A = \begin{pmatrix}
            0.995 & 0.005 \\
            0 & 0.98
        \end{pmatrix}
    \end{equation}
    and $
    m_w = (0, 0)^\top, 
    \Sigma_w = 0.07 I_2, 
    m_0 = (-0.4, -0.4)^\top, 
    \Sigma_0 = 0.3 I_2$, with $I_2$ the 2-dimensional identity matrix. Since $A$ is stable, this system converges in total variation to a unique invariant distribution $\mu = \mathcal{N}(m_{\mu}, \Sigma_\mu)$, where $m_\mu = (0, 0)^\top$ and 
    \begin{equation}
        \Sigma_\mu \approx \begin{pmatrix}
            7.36896  & 0.347856 \\
            0.347856 & 1.76768
        \end{pmatrix}.
    \end{equation}
    The latter was computed by solving the Riccati equation $\Sigma_\mu = A \Sigma_\mu A^\top + \Sigma_w$ with the \texttt{MatrixEquations.jl} package.\footnote{See \url{https://github.com/andreasvarga/MatrixEquations.jl}.}    
    $\hfill \blacksquare$
\end{example}

In the experiments below, the matrices $\mathbf{P}_\ell$ have been computed with one very long trajectory $\{y_i\}_{i = 1, \dots, 10^5}$ (see Remark~\ref{rem:computing steady state}), and the initial vectors $\mathbf{v}_{0}^\ell$ have been computed with $10^5 / \ell$ samples of length $\ell$. More details about how $\mathsf{TV}(\lambda_k, \tilde{\lambda}_{\ell, k})$ has been computed in practice can be found in Appendix~\ref{app:more_details}.

\textbf{Case 1 - Partially observable systems.} In this case, the system is only partially observable, and we only have access to the outputs. The state space is discretized as follows: each dimension is partitioned into $(-\infty,-1)$, $(1,\infty)$, and the interval $[-1,1]$ is further partitioned into $p$ subintervals of equal size. Thus the partition contains $n=(p+2)^2$ cells. In this case, we fix $p = 3$, leading to a 25 cells partition, and we approximate the discrete process $(Y_k)_{k \geq 0}$ with the process $(\tilde{Y}_{\ell, k})$ as defined in Section~\ref{sec:abstraction}, for $\ell = 1, 2, 3$. More precisely, we compute $\tilde{v}_{\ell, k}$ with Algorithm~\ref{alg:l-approx}, and compute $\mathsf{TV}(\lambda_k, \tilde{\lambda}_{\ell, k})$ for $k \in \{0, \dots, 100\}$. 

The results are in Figure~\ref{fig:motivation_1}. One can see that increasing memory reduces $\mathsf{TV}(\lambda_k, \tilde{\lambda}_{\ell, k})$ for most horizons $k$, \textbf{thereby increasing the approximation quality}. Moreover, one can see that the observed bounds follow the theoretical setting of Theorem~\ref{thm:increasing} and Theorem~\ref{thm:decreasing}, as the bounds seem to follow two regimes, first increasing and then decreasing. 

\begin{figure}[ht!]
    \centering
    \includegraphics[width=0.6\linewidth]{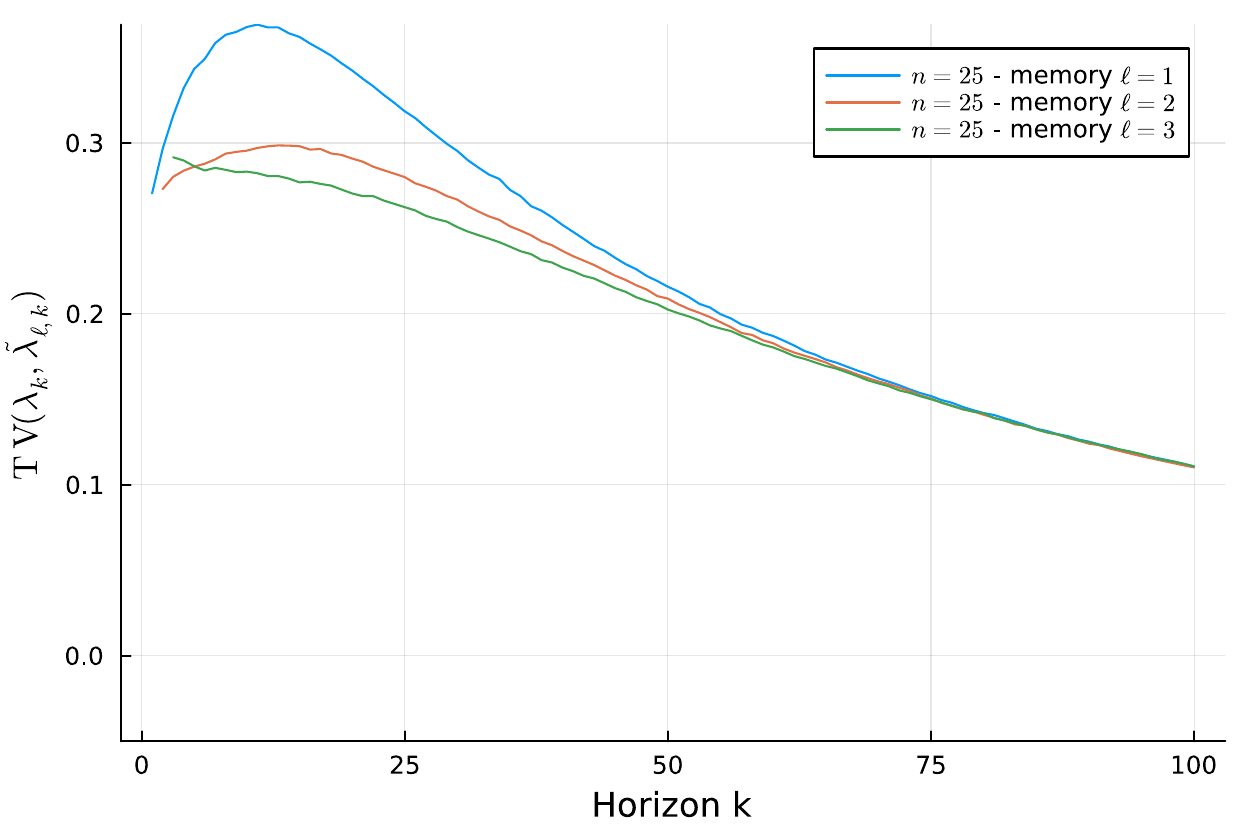}
    \caption{Approximation quality of $\ell$-memory Markov models for partially observable systems. Introducing memory improves approximation quality.}
    \label{fig:motivation_1}
\end{figure}

\textbf{Case 2 - Finite abstractions.} In this case, the state process of the system is available, but the state space is discretized in order to abstract the system. We consider two comparable settings: 
\begin{enumerate}
    \item (Classical Markov chain) The state space is discretized in $n = 729$ blocks in the same fashion as above (with $p = 25$). Memory $\ell=1$ is considered, leading to $729^2 = 531441$ transition probabilities $(\mathbf{P}_1)_{i_1, i_2}$ with $i_1, i_2 \in \{1, \dots, 729\}$.
    \item (2-memory Markov model) The state space is discretized in the same way (uniformly in the square $[-1, 1]^2$) in $n = 81$ blocks ($p=7$). Memory $\ell=2$ is considered, also leading to $81^3 = 531441$ transition probabilities $(\mathbf{P})_{i_1 i_2, i_2 i_3}$ with $i_1, i_2, i_3 \in \{1, \dots, 81\}$.
\end{enumerate}
The two settings lead to discrete objects of the same size, since one only needs to store $531441$ values to save them. For these two settings, we compute $\tilde{v}_{\ell, k}$ with Algorithm~\ref{alg:l-approx}, and compute $\mathsf{TV}(\lambda_k, \tilde{\lambda}_{\ell, k})$ for $k \in \{0, \dots, 100\}$. The results are in Figure~\ref{fig:motivation_2}. We observe that, even though the initial partition is coarser, larger memory leads to a better approximation, showcasing the fact that memory allows to construct \textbf{smarter abstractions than classical approaches}. 

\begin{figure}[ht!]
    \centering
    \includegraphics[width=0.6\linewidth]{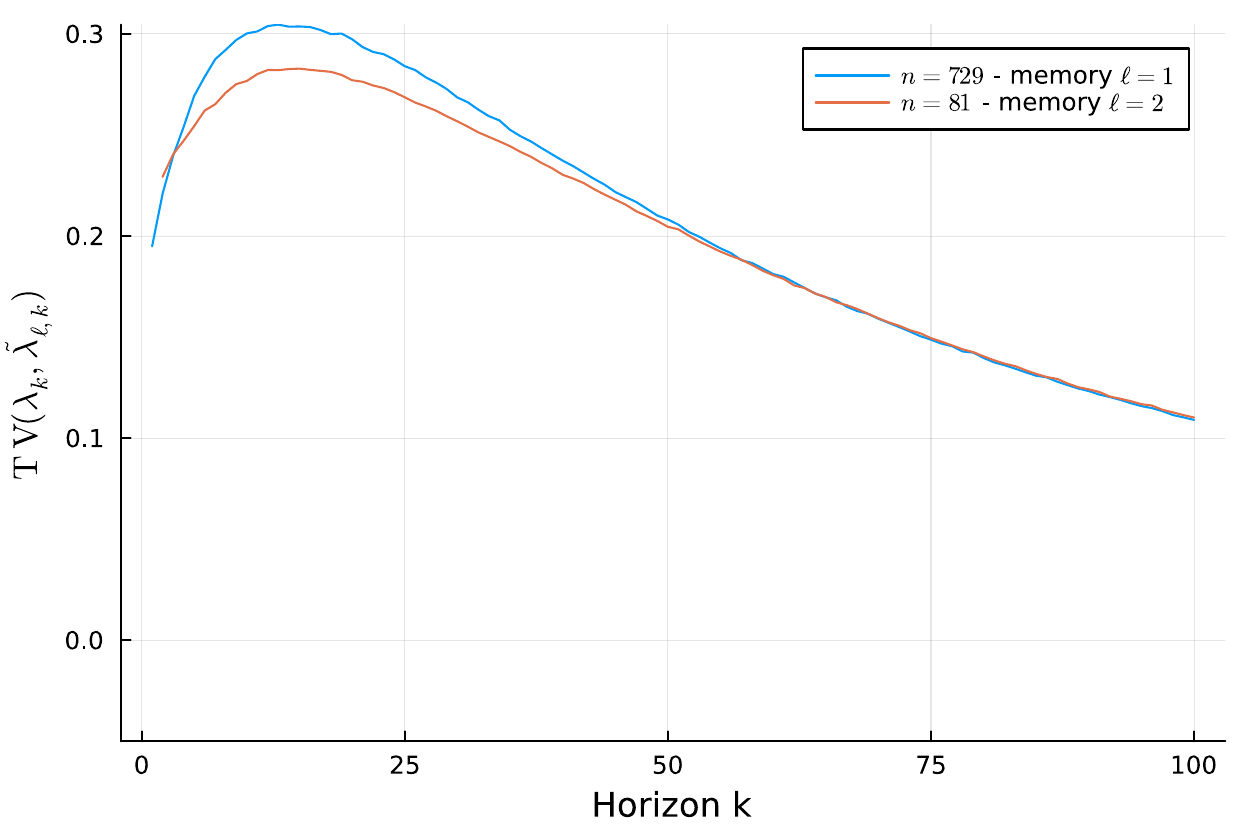}
    \caption{Approximation quality of memory-dependent abstractions of an observable system. One can see that, for the same number of transition probabilities ($n^{\ell+1} = 531441$ in both cases), starting from a coarser partition and increasing the memory leads to better approximations.}
    \label{fig:motivation_2}
\end{figure}

\section{Conclusions and further work} \label{sec:conclu}

In summary, in this work, we have introduced memory-dependent abstractions for stochastic systems. Our formalism, based on Galerkin approximations of lifted transfer operators, provides a theoretical framework for studying these abstractions. We have also upper bounded the approximation error, that we define as the total variation distance between the true distribution on the state space and the one of the memory-dependent approximation. We showed that this error consists of two regimes, one increasing (because of the accumulation of projection errors), and one decreasing (thanks to ergodicity). Through numerical experiments, we have demonstrated that increasing memory reduces the approximation error in various scenarios, highlighting how memory-dependent abstractions effectively address the issue of non-Markovianity of the discrete process induced by the discretization.

There are many interesting directions for future work. First, our numerical experiments suggest that for partially observable systems with a fixed partition, increasing memory allows to improve the approximation quality. Identifying the class of systems for which increasing memory guarantees a better approximation is an interesting direction for further research. Second, we plan to extend this work to the data-driven setting, and exploring aspects such as sample complexity as a function of the number of blocks and memory. Third, while our current bounds are valid, they suffer from conservatism and often exceed 1, the maximal value of any total variation distance. Therefore, we aim to investigate alternative approaches that directly rely on intermediate results on the 1-norm, rather than relying on 2-norm results as we do here.
    
    \bibliographystyle{apalike}
    \bibliography{ref}

\begin{thebibliography}{}

\bibitem[Abate et~al., 2010]{Abate2010}
Abate, A., Katoen, J.-P., Lygeros, J., and Prandini, M. (2010).
\newblock Approximate model checking of stochastic hybrid systems.
\newblock {\em European Journal of Control}, 16(6):624--641.

\bibitem[Alamir, 2022]{Alamir2022}
Alamir, M. (2022).
\newblock Learning against uncertainty in control engineering.
\newblock {\em Annual Reviews in Control}, 53:19--29.

\bibitem[Banse et~al., 2023a]{Banse2023}
Banse, A., Romao, L., Abate, A., and Jungers, R. (2023a).
\newblock Data-driven memory-dependent abstractions of dynamical systems.
\newblock In Matni, N., Morari, M., and Pappas, G.~J., editors, {\em
  Proceedings of The 5th Annual Learning for Dynamics and Control Conference},
  volume 211 of {\em Proceedings of Machine Learning Research}, pages 891--902.
  PMLR.

\bibitem[Banse et~al., 2023b]{Banse2023a}
Banse, A., Romao, L., Abate, A., and Jungers, R.~M. (2023b).
\newblock Data-driven abstractions via adaptive refinements and a kantorovich
  metric.
\newblock In {\em 2023 62nd IEEE Conference on Decision and Control (CDC)}.
  IEEE.

\bibitem[Delimpaltadakis et~al., 2024]{Delimpaltadakis2024}
Delimpaltadakis, G., Laurenti, L., and Mazo, M. (2024).
\newblock Formal analysis of the sampling behavior of stochastic
  event-triggered control.
\newblock {\em IEEE Transactions on Automatic Control}, 69(7):4491--4505.

\bibitem[Figueiredo et~al., 2024]{Figueiredo2024}
Figueiredo, E., Patane, A., Lahijanian, M., and Laurenti, L. (2024).
\newblock Uncertainty propagation in stochastic systems via mixture models with
  error quantification.

\bibitem[Guo et~al., 2023]{guo2023sampleefficientlearningpomdpsmultiple}
Guo, J., Chen, M., Wang, H., Xiong, C., Wang, M., and Bai, Y. (2023).
\newblock Sample-efficient learning of pomdps with multiple observations in
  hindsight.

\bibitem[Hairer, 2006]{Hairer2006ErgodicPO}
Hairer, M. (2006).
\newblock Ergodic properties of markov processes.

\bibitem[Katok and Hasselblatt, 1995]{Katok1995}
Katok, A. and Hasselblatt, B. (1995).
\newblock {\em Introduction to the Modern Theory of Dynamical Systems}.
\newblock Cambridge University Press.

\bibitem[Knight, 2002]{Knight2002}
Knight, J.~C. (2002).
\newblock Safety critical systems: challenges and directions.
\newblock In {\em Proceedings of the 24th international conference on Software
  engineering - ICSE ’02}, ICSE ’02, page 547. ACM Press.

\bibitem[Kontoyiannis and Meyn, 2011]{Kontoyiannis2011}
Kontoyiannis, I. and Meyn, S.~P. (2011).
\newblock Geometric ergodicity and the spectral gap of non-reversible markov
  chains.
\newblock {\em Probability Theory and Related Fields}, 154(1–2):327--339.

\bibitem[Lahijanian et~al., 2015]{Lahijanian2015}
Lahijanian, M., Andersson, S.~B., and Belta, C. (2015).
\newblock Formal verification and synthesis for discrete-time stochastic
  systems.
\newblock {\em IEEE Transactions on Automatic Control}, 60(8):2031--2045.

\bibitem[Lavaei et~al., 2022]{Lavaei2022}
Lavaei, A., Soudjani, S., Abate, A., and Zamani, M. (2022).
\newblock Automated verification and synthesis of stochastic hybrid systems: A
  survey.
\newblock {\em Automatica}, 146:110617.

\bibitem[Lavaei et~al., 2020]{Lavaei2020}
Lavaei, A., Soudjani, S., and Zamani, M. (2020).
\newblock Compositional abstraction of large-scale stochastic systems: A
  relaxed dissipativity approach.
\newblock {\em Nonlinear Analysis: Hybrid Systems}, 36:100880.

\bibitem[Lee and Seshia, 2016]{Lee2016}
Lee, E.~A. and Seshia, S.~A. (2016).
\newblock {\em Introduction to Embedded Systems A Cyber-Physical Systems
  Approach}.
\newblock MIT Press.

\bibitem[Lind and Marcus, 1995]{Lind1995}
Lind, D. and Marcus, B. (1995).
\newblock {\em An Introduction to Symbolic Dynamics and Coding}.
\newblock Cambridge University Press.

\bibitem[Meng and Liu, 2023]{Meng2023}
Meng, Y. and Liu, J. (2023).
\newblock Robustly complete finite-state abstractions for control synthesis of
  stochastic systems.
\newblock {\em IEEE Open Journal of Control Systems}, 2:235--248.

\bibitem[Nielsen et~al., 2013]{NielsenFackeldeyWeber2013}
Nielsen, A., Fackeldey, K., and Weber, M. (2013).
\newblock On a generalized transfer operator.
\newblock Technical Report 13-74, ZIB, Takustr. 7, 14195 Berlin.

\bibitem[Pairet et~al., 2022]{Pairet2022}
Pairet, E., Hernandez, J.~D., Carreras, M., Petillot, Y., and Lahijanian, M.
  (2022).
\newblock Online mapping and motion planning under uncertainty for safe
  navigation in unknown environments.
\newblock {\em IEEE Transactions on Automation Science and Engineering},
  19(4):3356--3378.

\bibitem[Prinz et~al., 2011]{Prinz2011}
Prinz, J.-H., Wu, H., Sarich, M., Keller, B., Senne, M., Held, M., Chodera,
  J.~D., Schütte, C., and Noé, F. (2011).
\newblock Markov models of molecular kinetics: Generation and validation.
\newblock {\em The Journal of Chemical Physics}, 134(17).

\bibitem[Sarich et~al., 2010]{Sarich2010}
Sarich, M., Noé, F., and Schütte, C. (2010).
\newblock On the approximation quality of markov state models.
\newblock {\em Multiscale Modeling and Simulation}, 8(4):1154--1177.

\bibitem[Schmuck and Raisch, 2014]{Schmuck2014}
Schmuck, A.-K. and Raisch, J. (2014).
\newblock Asynchronous l-complete approximations.
\newblock {\em Systems and Control Letters}, 73:67--75.

\bibitem[Schmuck et~al., 2015]{Schmuck2015}
Schmuck, A.-K., Tabuada, P., and Raisch, J. (2015).
\newblock Comparing asynchronous l-complete approximations and quotient based
  abstractions.
\newblock In {\em 2015 54th IEEE Conference on Decision and Control (CDC)},
  pages 6823--6829. IEEE.

\bibitem[Schütte et~al., 2001]{Schuette2001}
Schütte, C., Huisinga, W., and Deuflhard, P. (2001).
\newblock {\em Transfer Operator Approach to Conformational Dynamics in
  Biomolecular Systems}, pages 191--223.
\newblock Springer Berlin Heidelberg.

\bibitem[Shiryaev, 2016]{Shiryaev2016}
Shiryaev, A.~N. (2016).
\newblock {\em Probability-1: Volume 1}.
\newblock Springer New York.

\bibitem[Suilen et~al., 2025]{Suilen2025}
Suilen, M., Badings, T., Bovy, E.~M., Parker, D., and Jansen, N. (2025).
\newblock {\em Robust Markov Decision Processes: A Place Where AI and Formal
  Methods Meet}, pages 126--154.
\newblock Springer Nature Switzerland, Cham.

\bibitem[Wu and Chu, 2017]{Wu2017}
Wu, S.-J. and Chu, M.~T. (2017).
\newblock Markov chains with memory, tensor formulation, and the dynamics of
  power iteration.
\newblock {\em Applied Mathematics and Computation}, 303:226--239.

\bibitem[Zhou et~al., 2010]{5404339}
Zhou, E., Fu, M.~C., and Marcus, S.~I. (2010).
\newblock Solving continuous-state pomdps via density projection.
\newblock {\em IEEE Transactions on Automatic Control}, 55(5):1101--1116.

\end{thebibliography}

    \appendix
\section{Proofs} \label{sec:proofs}

Our results rely on the Hölder's inequality, that we recall hereinafter.

\begin{lemma}[Hölder's inequality]
    Given a measurable space $(E, \mathcal{F})$, together with a measure $\mu$, and $p, q \in [0, + \infty]$ such that $1/p + 1/q = 1$.  Then, for all $f : E \to \mathbb{R}$ and $g : E \to \mathbb{R}$, it holds that
    \begin{equation*}
        \|fg\|_1 \leq \|f\|_p \|g\|_q.
    \end{equation*}
\end{lemma}
A direct consequence of Hölder's inequality is that, for all function $f \in L^2(\mu)$, it holds that
\begin{equation*}
    \|f\|_1 \leq \|f\|_2.
\end{equation*}

Our proofs also rely on the following lemma, that holds under Assumption~\ref{ass:spectral}.
\begin{lemma}[{\cite[Lemma~2.2]{Sarich2010}}] \label{lemma:sarich}
    For all initial densities $v_0^\ell \in L^2(\mu^\ell)$, and $k \geq \ell - 1$ it holds that 
    \begin{equation*}
        \|P_\ell^{k-\ell+1} v_0^\ell - \mathds{1}\|_2 \leq \|(T_\ell Q_\ell)^{k-\ell+1}v_0^\ell - \mathds{1}\|_2 \leq e_{\ell, 1}^{k-\ell+1}\|v_0^\ell\|_2.
    \end{equation*}
\end{lemma}

Finally, we will also need the fact that all transfer operator has a unitary norm.
\begin{lemma} \label{lemma:unitary_transfer}
    Given a state space $E$ and a kernel $\tau$, the transfer operator $T : L^2(\mu) \to L^2(\mu)$ is such that $\|T\|_1 = 1$.
\end{lemma}
\begin{proof}
    We first recall the operator norm
    \begin{equation*}
        \|T\|_1 = \sup_{f \in L^2(\mu) : \|f\|_1 \leq 1} \|Tf\|_1.
    \end{equation*}
    First we prove that $\|T\|_1 \geq 1$. Take any nonnegative function $f : E \to \mathbb{R}_{\geq 0}$ such that $\|f\|_1 = 1$. Then it holds that 
    \begin{equation*}
        \|Tf\|_1 = \int_{x \in E} |(Tf)(x)| \mu(\text dx).
    \end{equation*}
    By definition of the transfer operator, if $f$ is nonnegative, then $Tf$ is also nonnegative. Therefore, 
    \begin{equation*}
    \begin{aligned}
         \|Tf\|_1 &= \int_{x \in E} (Tf)(x) \mu(\text dx)  \\
         &= \int_{x \in E} f(x) \mu(\text dx) \\
         &= 1, 
    \end{aligned}
    \end{equation*}
    where the first equality holds by \eqref{eq:transfer}, and the second by assumption. By definition of $\sup$, this proves the first claim.

    Now we prove that $\|T\|_1 \leq 1$. Any function $f : E \to \mathbb{R}$ can be written as $f = f^+ - f^-$, where $f^+ : E \to \mathbb{R}_{\geq 0}$ and $f^- : E \to \mathbb{R}_{\geq 0}$ are respectively the positive and negative parts of $f$. Therefore, for all functions $f$ such that $\|f\|_1 \leq 1$, it holds that 
    \begin{equation*}
    \begin{aligned}
        \|Tf\|_1 &= \int_{x \in E} |(Tf)(x)| \mu(\text dx) \\
        &= \int_{x \in E} |(T(f^+ - f^-))(x)| \mu(\text dx) \\
        &= \int_{x \in E} |(Tf^+)(x) - (Tf^-)(x)| \mu(\text dx) \\
        &= \int_{x \in E} (Tf^+)(x) + (Tf^-)(x) \mu(\text dx) \\
        &= \int_{x \in E} (Tf^+)(x)\mu(\text dx)  + \int_{x \in E} (Tf^-)(x) \mu(\text dx) \\
        &= \int_{x \in E} f^+(x)\mu(\text dx)  + \int_{x \in E} f^-(x) \mu(\text dx) \\
        &= \int_{x \in E} f^+(x) + f^-(x) \mu(\text dx) \\
        &= \int_{x \in E} |f^+(x) - f^-(x)| \mu(\text dx) \\
        &=\int_{x \in E} |f(x)| \mu(\text dx) \\
        &\leq 1, 
    \end{aligned}
    \end{equation*}
    which concludes the proof.
\end{proof}

\subsection{Proof of Proposition~\ref{prop:galerkin}}

First we define an orthonormal basis for $D_n^\ell$, given by $\{\phi_{i_1 \dots i_\ell}\}_{i_1, \dots, i_\ell \in [n]}$: 
\begin{equation*}
    \phi_{i_1 \dots i_\ell}(x_1, \dots, x_{\ell}) = \frac{\chi_{A_{i_1}}(x_1) \dots \chi_{A_{i_\ell}}(x_\ell)}{\sqrt{\mathbb{P}[X^\mu_1 \in A_{i_1}, \dots, X^\mu_{\ell} \in A_{i_\ell}]}}. 
\end{equation*}
The orthogonality of the basis follows, since for all $i_1, \dots, i_\ell \in [n]$,
\begin{equation*}
\begin{aligned}
    \langle \phi_{i_1 \dots i_\ell}, \phi_{i_1 \dots i_\ell} \rangle &= \frac{\int_{x_1 \in A_{i_1}} \dots\int_{x_\ell \in A_{i_\ell}} \tau(\text dx_{\ell} | x_{\ell-1}) \dots \tau(\text dx_1 | x_0) \mu(\text dx_0)}{\mathbb{P}[X^\mu_1 \in A_{i_1}, \dots, X^\mu_{\ell} \in A_{i_\ell}]}= 1, 
\end{aligned}
\end{equation*}
and for all $i_1 \dots i_\ell \neq j_1 \dots j_\ell$, $\langle \phi_{i_1 \dots i_\ell}, \phi_{j_1 \dots j_\ell} \rangle = 0$. As a consequence, by definition of $Q_\ell$, it holds that
\begin{equation} \label{eq:proj_orthonormal}
    Q_\ell v^\ell = \sum_{i_1 \dots i_\ell} \langle v^\ell, \phi_{i_1 \dots i_\ell} \rangle \phi_{i_1 \dots i_\ell}
\end{equation}
for all $v^\ell \in L^2(\mu^\ell)$.

Now, by definition of $P_\ell$ and by \eqref{eq:proj_orthonormal}, for all $i_0, \dots, i_{\ell-1} \in [n]$,
\begin{equation} \label{eq:proof_prop_1}
\begin{aligned} 
    P_\ell \psi_{i_0, \dots, i_{\ell-1}}
    &= Q_\ell T_\ell Q_\ell \psi_{i_0, \dots, i_{\ell-1}} \\
    &= Q_\ell \left(T_\ell \psi_{i_0, \dots, i_{\ell-1}}\right) \\
    &= \sum_{j_1, \dots, j_\ell \in [n]} \langle T_\ell \psi_{i_0, \dots, i_{\ell-1}}, \phi_{j_1, \dots, j_{\ell}} \rangle \phi_{j_1, \dots, j_{\ell}} \\
    &= \sum_{j_1, \dots, j_\ell \in [n]} \frac{\left\langle T_\ell \left(\chi_{A_{i_0}} \dots \chi_{A_{i_{\ell-1}}}\right), \chi_{A_{j_1}} \dots \chi_{A_{j_{\ell}}} \right\rangle}{\mathbb{P}[X^\mu_0 \in A_{i_0}, \dots, X^\mu_{\ell-1} \in A_{i_{\ell-1}}]} \psi_{j_1, \dots, j_{\ell}}.
\end{aligned}
\end{equation}
By~\eqref{eq:kernel_lift} , the dot product above equals zero if $j_1 \dots j_{\ell-1} \neq i_1 \dots i_{\ell-1}$, and, otherwise, is equal to 
\begin{equation*}
\begin{aligned}
    &\hspace{-0.2cm}\int_{x_0 \in A_{i_0}} \dots \int_{x_{\ell-1} \in A_{i_{\ell-1}}} \tau(A_{j_\ell} | x_{\ell-1})\tau(dx_{\ell-1} | x_{\ell-2}) \dots \tau(\text dx_1 | x_0) \mu(\text dx_0)\\
    &\quad = \mathbb{P}[X_0^\mu \in A_{i_0}, \dots, X_{\ell}^\mu \in A_{j_\ell}].
\end{aligned}
\end{equation*}
Therefore, inserting this in \eqref{eq:proof_prop_1} yields

\begin{equation*}
    P_\ell \psi_{j_0, \dots, j_{\ell-1}} = \sum_{j_1, \dots, j_\ell \in [n]} 
    \frac{\mathbb{P}[X_0^\mu \in A_{i_0}, \dots, X_{\ell}^\mu \in A_{j_\ell}]}{\mathbb{P}[X^\mu_0 \in A_{i_0}, \dots, X^\mu_{\ell-1} \in A_{i_{\ell-1}}]} \psi_{j_1, \dots, j_{\ell}}, 
\end{equation*}
which is \eqref{eq:prop1_to_prove} by definition of the output process, and the proof is completed. $\hfill \square$

\subsection{Proof of Theorem~\ref{thm:increasing}}

First, by the triangular inequality, it holds that 
\begin{equation*}
    \mathsf{TV}(\lambda^\ell_{k-\ell+1}, \tilde{\lambda}^\ell_{k-\ell+1}) \leq \mathsf{TV}(\lambda^\ell_{k-\ell+1}, \overline{\lambda}^\ell_{k-\ell+1}) + \mathsf{TV}(\overline{\lambda}^\ell_{k-\ell+1}, \tilde{\lambda}^\ell_{k-\ell+1}), 
\end{equation*}
where $\lambda^\ell_{k-\ell+1}$, $\overline{\lambda}^\ell_{k-\ell+1}$ and $\tilde{\lambda}^\ell_{k-\ell+1}$ are respectively the measures corresponding to 
\begin{equation*}
    v^\ell_{k-\ell+1} = T_\ell^{k-\ell+1}v^\ell_0, 
\end{equation*}
\begin{equation*}
    \overline{v}^\ell_{k-\ell+1} = T_\ell P_\ell^{k-\ell} v^\ell_0, 
\end{equation*}
and
\begin{equation*}
    \tilde{v}^\ell_{k-\ell+1} = P_\ell^{k-\ell+1}v^\ell_0.  
\end{equation*}
First, we show that 
\begin{equation*}
    \mathsf{TV}(\lambda^\ell_{k-\ell+1}, \overline{\lambda}^\ell_{k-\ell+1}) \leq \mathsf{TV}(\lambda^\ell_{k-\ell}, \tilde{\lambda}^\ell_{k-\ell}).
\end{equation*}
It holds that $v^\ell_{k-\ell+1} = T_\ell v^\ell_{k-\ell}$, and that $\overline{v}^\ell_{k-\ell+1} = T_\ell \tilde{v}^\ell_{k-\ell}$. Therefore, by definition of the total variation distance and by Lemma~\ref{lemma:unitary_transfer},
\begin{equation*}
\begin{aligned}
    \mathsf{TV}(\lambda^\ell_{k-\ell+1}, \overline{\lambda}^\ell_{k-\ell+1}) &= \frac{1}{2} \left\|T_\ell \left(v^\ell_{k-\ell} - \tilde{v}^{\ell}_{k-\ell} \right)\right\|_1 \\
    &\leq \frac{1}{2} \|T_\ell\|_1 \|v^\ell_{k-\ell} - \tilde{v}^\ell_{k-\ell}\|_1 \\
    &= \frac{1}{2} \|v^\ell_{k-\ell} - \tilde{v}^\ell_{k-\ell}\|_1 \\
    &= \mathsf{TV}(\lambda^\ell_{k-\ell}, \tilde{\lambda}^\ell_{k-\ell}).
\end{aligned}
\end{equation*} 
Second, we show that 
\begin{equation*}
    \mathsf{TV}(\overline{\lambda}^\ell_{k-\ell+1}, \tilde{\lambda}^\ell_{k-\ell+1}) \leq \frac{1}{2} (me_{\ell, 1}\delta_\ell + r_\ell) e_{\ell, 1}^{k-\ell}\|v_0^\ell\|_2.
\end{equation*}
It holds that $\overline{v}^\ell_{k-\ell+1} = T_\ell \tilde{v}_{k-\ell}^\ell$, and that $\tilde{v}^\ell_{k-\ell+1} = P_\ell \tilde{v}^\ell_{k-\ell} = Q_\ell T_\ell  \tilde{v}^\ell_{k-\ell}$. Therefore, by definition of the total variation distance, and by Hölder's inequality, it holds that 
\begin{equation*}
\begin{aligned}
    \mathsf{TV}(\overline{\lambda}^\ell_{k-\ell+1}, \tilde{\lambda}^\ell_{k-\ell+1}) &= \frac{1}{2}\| (T_\ell - Q_\ell T_\ell)\tilde{v}^\ell_{k-\ell} \|_1 \leq \frac{1}{2} \| (T_\ell - Q_\ell T_\ell)\tilde{v}^\ell_{k-\ell} \|_2.
\end{aligned}
\end{equation*}
Now we follow a similar reasoning as in the proof of \cite[Theorem~3.1]{Sarich2010}. Let $Q^\perp_{\ell} = \text{Id} - Q_\ell$, where $\text{Id}$ the identity operator. Since $Q^\perp_\ell \mathds{1} = 0$, it holds that
\begin{equation*}
\begin{aligned}
    \| (T_\ell - Q_\ell T_\ell) \tilde{v}^\ell_{k-\ell} \|_2 &= \| Q^\perp_\ell T_\ell \tilde{v}^\ell_{k-\ell} \|_2 \\
    &= \|Q_\ell^\perp T_\ell Q_\ell (T_\ell Q_\ell)^{k-\ell} \tilde{v}^\ell_{0} \|_2 \\
    &= \|Q_\ell^\perp T_\ell Q_\ell ((T_\ell Q_\ell)^{k-\ell} \tilde{v}^\ell_{0} - \mathds{1}) \|_2 \\
    &\leq \|Q_\ell^\perp T_\ell Q_\ell\|_2 \| (T_\ell Q_\ell)^{k-\ell} \tilde{v}^\ell_{0} - \mathds{1} \|_2.
\end{aligned}
\end{equation*}
Following \cite[Eq. (63)]{Sarich2010}, the first factor is such that 
\begin{equation*}
    \|Q_\ell^\perp T_\ell Q_\ell\|_2 \leq me_{\ell, 1} \delta_\ell, 
\end{equation*}
and, by Lemma~\ref{lemma:sarich}, the second is such that 
\begin{equation*}
    \| (T_\ell Q_\ell)^{k-\ell} \tilde{v}^\ell_{0} - \mathds{1} \|_2 \leq e_{\ell, 1}^{k-\ell}\|v_0^\ell\|_2, 
\end{equation*}
which concludes the proof. $\hfill \square$

\subsection{Proof of Theorem~\ref{thm:decreasing}}

By the triangular inequality, it holds that
\begin{equation*}
    \mathsf{TV}(\lambda^\ell_{k - \ell + 1}, \tilde{\lambda}^\ell_{k - \ell + 1}) \leq \mathsf{TV}(\lambda^\ell_{k - \ell + 1}, \mu^\ell) + \mathsf{TV}(\mu^\ell, \tilde{\lambda}^\ell_{k - \ell + 1}), 
\end{equation*}
where $\mu^\ell$ is the lifted invariant measure as defined in \eqref{eq:invariant_lift}. By Hölder's inequality and by Assumption~\ref{ass:spectral}, it holds that 
\begin{equation*}
    \mathsf{TV}(\lambda^\ell_{k - \ell + 1}, \mu^\ell) \leq \frac{1}{2}\|v^\ell_{k - \ell + 1} - \mathds{1}\|_2
    \leq \frac{1}{2}e_{\ell, 1}^{k-\ell+1} \|v_0^\ell\|_2.
\end{equation*}
By Lemma~\ref{lemma:sarich}, we can follow the exact same reasoning as above to get
\begin{equation*}
    \mathsf{TV}(\mu^\ell, \tilde{\lambda}^\ell_{k - \ell + 1})
    \leq \frac{1}{2}e_{\ell, 1}^{k-\ell+1} \|v_0^\ell\|_2, 
\end{equation*}
which concludes the proof. $\hfill \square$

\subsection{Proof of Corollary~\ref{cor:final_bound}}

First, by Theorem~\ref{thm:increasing}, it holds that 
\begin{equation*}
\begin{aligned}
    \mathsf{TV}(\lambda^\ell_{k-\ell+1}, \tilde{\lambda}^\ell_{k-\ell+1})
    &\leq \mathsf{TV}(\lambda^\ell_0, \tilde{\lambda}^\ell_0) + \sum_{i = 0}^{k-\ell} \frac{1}{2}\left(m e_{\ell, 1} \delta_\ell + r_\ell\right) e_{\ell, 1}^{i} \|v_0^\ell\|_2 \\
    &= \mathsf{TV}(\lambda^\ell_0, \tilde{\lambda}^\ell_0) + \frac{1}{2}\left(m e_{\ell, 1} \delta_\ell + r_\ell\right) \left(\sum_{i = 0}^{k-\ell} e_{\ell, 1}^{i}\right) \|v_0^\ell\|_2 \\
    &= \mathsf{TV}(\lambda^\ell_0, \tilde{\lambda}^\ell_0) + \frac{1}{2}\left(m e_{\ell, 1} \delta_\ell + r_\ell\right) \frac{1-e_{\ell, 1}^{k-\ell+1}}{1- e_{\ell, 1}} \|v_0^\ell\|_2.
\end{aligned}
\end{equation*}
Now it remains to show that 
\begin{equation*}
    \mathsf{TV}(\lambda_{k}, \tilde{\lambda}_{\ell, k}) \leq \mathsf{TV}(\lambda^\ell_{k-\ell+1}, \tilde{\lambda}^\ell_{k-\ell+1}).
\end{equation*}
By definition of the total variation distance, 
\begin{equation*}
\begin{aligned}
    \mathsf{TV}(\lambda^\ell_{k-\ell+1}, \tilde{\lambda}^\ell_{k-\ell+1}) &= \sup_{A_{k-\ell+1}, \dots, A_{k} \in \mathcal{B}(E)} \left|
    \begin{aligned}
        &\lambda^\ell_{k-\ell+1}(A_{k-\ell+1} \times \dots \times A_k) \\
        &\quad - \lambda^\ell_{k-\ell+1}(A_{k-\ell+1} \times \dots \times A_k)
    \end{aligned} 
    \right| \\
    &\geq \sup_{A_k \in \mathcal{B}(E)} \left|
    \begin{aligned}
        &\lambda^\ell_{k-\ell+1}(E \times \dots \times E \times A_k) \\
        &\quad - \lambda^\ell_{k-\ell+1}(E \times \dots \times E \times A_k)
    \end{aligned} 
    \right| \\
    &= \sup_{A_k \in \mathcal{B}(E)} |\lambda_k(A_k) - \lambda_{\ell, k}(A_k)| \\
    &= \mathsf{TV}(\lambda_k,  \lambda_{\ell, k}),
\end{aligned}
\end{equation*}
and the proof is completed. $\hfill \square$

\section{Details about the numerical experiments} \label{app:more_details}

In order to compute $\mathsf{TV}(\lambda_k, \tilde{\lambda}_{\ell, k})$ from Example~\ref{ex:2d} and $\tilde{v}_{\ell, k}$, we proceed as follows. First, we know that 
\begin{equation*}
    \lambda_{k} = \mathcal{N}(m_k, \Sigma_k), 
\end{equation*}
where $m_k$ and $\Sigma_k$ satisfy the recurrence
\begin{equation*}
    m_{k+1} = A m_k, \quad \Sigma_{k+1} = A \Sigma_k A^\top + \Sigma_w.
\end{equation*}
The $\mu$-weighted probability density function $v_k \in L^2(\mu)$ is given by 
\begin{equation*}
    v_k = \frac{\text d\lambda_k}{\text d\mu} = \frac{\text d\lambda_k}{\text d \lambda^*} \frac{\text d\lambda^*}{\text d \mu} = \frac{\text d\lambda_k}{\text d \lambda^*} \left(\frac{\text d \mu}{\text d\lambda^*} \right)^{-1}, 
\end{equation*}
where $\lambda^*$ is the Lebesgue measure. Therefore, $v_k$ is defined as 
\begin{equation*}
    v_k(x) = \frac{f_{\mathcal{N}(m_k, \Sigma_k)}(x)}{f_{\mathcal{N}(m_\mu, \Sigma_\mu)}(x)}, 
\end{equation*}
where $f_{\mathcal{N}(m, \Sigma)}$ is the usual probability density function of Gaussian distributions with respect to the Lebesgue measure. Finally, we compute the total variation with a Monte-Carlo approximation, that is 
\begin{equation*}
\begin{aligned}
    \mathsf{TV}(\lambda_k, \tilde{\lambda}_{\ell, k}) &= \frac{1}{2} \int_{x \in E} |v_k(x) - \tilde{v}_{\ell, k}(x)| \mu(\text dx) \\
    &\approx \frac{1}{2} \sum_{i = 1}^{10^4} |v_k(x_i) - \tilde{v}_{\ell, k}(x_i)|, 
\end{aligned}
\end{equation*}
where $x_i$ are i.i.d. samples from the invariant measure. 
\end{document}